\documentclass[journal]{IEEEtran}
\usepackage{cite}

\ifCLASSINFOpdf
\else
\fi
\usepackage{graphicx}
\usepackage{tabularx}
\usepackage{caption}
\usepackage{subcaption}
\usepackage[dvipsnames]{xcolor}
\usepackage{cancel}
\usepackage{algpseudocode}
\usepackage{hyperref}
\usepackage{amsmath}
\usepackage{soul,color}
\usepackage{makecell}
\usepackage{array}
\usepackage{adjustbox}
\usepackage{makecell}
\usepackage{footnote}
\usepackage[table]{xcolor}
\definecolor{myrowgray}{gray}{0.92}
\usepackage{amssymb} 
\usepackage{booktabs}   
\renewcommand{\arraystretch}{1.2} 
\usepackage{amssymb} 
\usepackage{pifont}  
\newcommand{\xmark}{\ding{55}} 
\usepackage{footnote}
\usepackage{threeparttable}
\usepackage{array}
\newcolumntype{Y}{>{\centering\arraybackslash}X} 
\usepackage[table,xcdraw]{xcolor}

\usepackage[ruled,vlined]{algorithm2e}
\newtheorem{theorem}{Theorem}

\newtheorem{proof}{Proof}

\usepackage{textcomp}
\usepackage{subcaption}

\definecolor{myblue}{RGB}{65,105,225}
\definecolor{myred}{RGB}{178,34,34}
\definecolor{mygold}{RGB}{255,215,0}
\definecolor{mygreen}{RGB}{34,139,34}
\usepackage{microtype}

\usepackage{enumitem}
\setlist{leftmargin=*, itemsep=1pt, topsep=2pt, parsep=0pt, partopsep=0pt}

\renewcommand{\arraystretch}{1.08}
\usepackage{balance}

\def\BibTeX{{\rm B\kern-.05em{\sc i\kern-.025em b}\kern-.08em
    T\kern-.1667em\lower.7ex\hbox{E}\kern-.125emX}}

\makeatletter
\renewcommand{\section}{\@startsection{section}{1}{\z@}%
  {1.2ex plus 0.5ex minus 0.2ex}
  {0.7ex plus 0.2ex}
  {\normalfont\scshape\centering}}
\renewcommand{\subsection}{\@startsection{subsection}{2}{\z@}%
  {1.0ex plus 0.3ex minus 0.2ex}
  {0.5ex plus 0.2ex}
  {\normalfont\itshape}}
\makeatother

\begin{document}
%

\title{Split-Aware Function Placement with Availability Guarantees and Optical Provisioning in vRANs}
%
%
%

\author{Mayank Ramnani, Shasank Dixit,
         Sushil Yadav, Saad Ahmed
         and Sidharth Sharma 
\thanks{Mayank Ramnani, Sushil Yadav and Saad Ahmed are with the Department
of Computer Science and Engineering, Indian Institute of Technology Indore, India, 453552 (e-mail: phd2201101004@iiti.ac.in). 
Shasank Dixit and Sidharth Sharma are with the Department
of Computer Science and Engineering, Indian Institute of Technology Jodhpur, India, 342037 (e-mail: p25cs0008@iitj.ac.in; sidharth@iitj.ac.in).}}

%
%

\markboth{Journal of \LaTeX\ Class Files,~Vol.~14, No.~8, August~2015}%
{Shell \MakeLowercase{\textit{et al.}}: Bare Demo of IEEEtran.cls for IEEE Journals}
%

\bstctlcite{BSTcontrol}



\maketitle

\begin{abstract}
The rapid evolution of beyond-5G and emerging 6G networks is driving the need for flexible, reliable, and cost-efficient virtualized Radio Access Network (vRAN) architectures capable of supporting heterogeneous services such as enhanced Mobile Broadband (eMBB), Ultra-Reliable Low-Latency Communication (URLLC), and Massive Machine-Type Communication (mMTC).
Future disaggregated RAN systems are expected to rely heavily on network slicing, functional split flexibility, and optical x-haul infrastructures to support stringent performance, scalability, and availability requirements.
In this paper, we present an integrated framework for reliable, slice-aware, and functional split-aware Virtual Network Function (VNF) placement with lightpath provisioning in disaggregated vRAN environments.
The proposed approach maximizes mobile network operators’ profit by jointly optimizing function placement and optical resource allocation under latency, processing, bandwidth, and availability constraints. 
We formulate the problem as an Integer Linear Programming (ILP) model with two variants: one that employs unshared backups and another that uses a more cost-efficient shared backup scheme. 
To address ILP complexity, we develop a heuristic algorithm and a Genetic Algorithm (GA)-based metaheuristic that yields near-optimal solutions in real time. 
Extensive evaluations on topologies up to 128 nodes show that shared backup variants yield up to 18\% higher profit, while maintaining up to 5–10\% lower normalized CPU usage than unshared counterparts.
\end{abstract}


%
\IEEEpeerreviewmaketitle

\section{Introduction}
\label{sec:intro}
%
%
%
%
\IEEEPARstart{F}{or} the cellular network infrastructure, the Radio Access Network (RAN) is a crucial component that provides the interface between User Equipment (UE) and the core network. 
However, traditional, tightly integrated RAN architectures cannot deliver the flexibility, scalability, and efficiency needed to keep up with the rapid rise in mobile traffic and the wide variety of applications, from immersive XR to mission-critical IoT.
Therefore, this traditional architecture has evolved over the years towards greater programmability and flexibility. 

One major step in this evolution was the introduction of Centralized RAN (C-RAN), where baseband functions are pooled centrally, while the radios remain at the cell sites. This separation enabled more efficient resource sharing and better coordination of radio functions. 
Building on this, 3GPP’s Next-Generation RAN (NG-RAN) defined \textit{functional splitting} that takes RAN design a step further by enabling the radio protocol stack to be divided across three logical entities: Radio Units (RUs), Distributed Units (DUs), and Central Units (CUs).
The NG-RAN standard specifies eight possible split options \cite{larsen2018survey}, each describing a different way of distributing functions among these three entities, with distinct latency, bandwidth, and processing requirements. 
By enabling flexible combinations of centralization and distribution, functional splitting offers greater adaptability than C-RAN and is better suited to meet the diverse performance demands of next-generation services. 
However, selecting the optimal split is challenging; it depends on a variety of factors, including traffic demand, network topology, and the specific QoS requirements of diverse services \cite{harutyunyan2018flex5g,morais2022placeran}.

\par The amalgamation of Network Function Virtualization (NFV) into the RAN has given rise to the virtualized RAN (vRAN) paradigm \cite{morgado2018survey}, where CUs, DUs, and RUs operate as software instances on Commercial-Off-The-Shelf (COTS) hardware.  vRAN offers flexibility in deploying radio functions and managing resources, enabling more efficient and adaptable cellular networks. Supported by industry efforts such as the O-RAN Alliance, it also paves the way for cloud-native deployments and dynamic scaling. At the same time, it brings new operational challenges. Placement of virtualized radio functions becomes more complex because each functional split comes with its own transport and computing requirements. As a result, moving from a fully centralized RAN (C-RAN) to a distributed vRAN requires a careful balance between maintaining high performance and resource efficiency \cite{sehier2019transport} \cite{gstr2018transport}.

\par In 5G NG-RAN deployments, service diversity is central. 
Ultra-Reliable Low Latency Communication (URLLC) services require sub-1ms latency and five-nines availability for use cases such as telesurgery and cooperative driving. 
Enhanced Mobile Broadband (eMBB) services require multi-Gbps throughput for 8K/VR streaming, while Massive Machine-Type Communication (mMTC) services prioritize device density and energy efficiency \cite{8723481}. 
Network slicing \cite{zhang2019overview} enables operators to instantiate multiple logically isolated RAN instances over a shared physical infrastructure, each tailored to a specific service profile. 
vRAN’s software-defined nature makes such dynamic slice instantiation feasible \cite{ordonez2017network}, but also raises the bar for slice-aware functional placement. 
This is because the system must map each RU to a CU/DU configuration that meets slice-specific QoS while respecting compute and transport constraints.

\par Availability and reliability considerations for the slices further complicate this task. Telecom-grade hardware traditionally delivers “five-nines” availability \cite {9126778, 8847037}. However, COTS-based vRAN is not that robust. Failures in compute nodes, transport links, or optical equipment can breach slice SLAs, especially in URLLC slices. 

\par Compounding the challenge is the optical transport network that interconnects RUs, DUs, and CUs. 
In the widely used DWDM-based optical transport, wavelength assignment for lightpaths must meet per-slice latency and availability targets within tight spectral budgets.
Standards from ITU-T (e.g., G.709) and the Optical Internetworking Forum (OIF) emphasize coordinated radio-optical orchestration \cite{zhang2017demonstration}. 
However, most vRAN placement research models the transport as static and unconstrained, leading to suboptimal or infeasible end-to-end solutions \cite{8993833}, \cite{8722598}.

Prior works on vRAN function placement  \cite{garcia2018wizhaul, morais2022placeran ,9733024, sen2025slice,ahsan2023efficient,begic2024framework,li2024reliability} 
address either functional-split-aware optimization or transport resource constraints, or optical lightpath protection, but not joint optimization of (1) functional split selection, (2) slice-aware placement, (3) availability-aware redundancy, and (4) optical lightpath provisioning.
This gap is significant as in real NG-RAN deployments, these factors are interdependent and must be co-optimized to meet 5G/6G performance expectations.  We propose a comprehensive framework in this paper that integrates split-aware function placement, availability, slice-type requirements, and lightpath provisioning. Our key contributions are as follows:


\begin{enumerate}
\item \textbf{Joint ILP formulation.} We design an ILP that jointly: (i) selects the Virtual Network Configuration or VNC (functional split and RU/DU/CU mapping), (ii) places functions to meet slice QoS and availability, and (iii) provisions primary and backup lightpaths with wavelength assignment under bandwidth and capacity limits. Unlike prior work, split selection, placement, and DWDM provisioning are co-optimized in a single model.

\item \textbf{Availability models with protection choices.} We design two protection regimes: \emph{Dedicated Protection} (unshared backups per slice request) and \emph{Shared Protection} (backup resources shared across requests). The latter reduces redundancy and cost while satisfying per-slice availability constraints.

\item \textbf{Scalable solution methods.} To address ILP scalability, we (i) propose a greedy heuristic that, for each RU, selects an appropriate VNC and allocates suitable primary and backup path(s) with wavelength channels, and (ii) develop a Genetic Algorithm (GA) based metaheuristic that refines placements to balance runtime and solution quality.

\item \textbf{Reference system model with optical x-haul.} We adopt a flexible x-haul transport network where any non-RU node can act as a DU or CU. The transport is optical; paths from RU to CU are provisioned as lightpaths. 
The optimization operates directly on this architecture.
\end{enumerate}

The remainder of the paper is organized as follows. Section II surveys related work on 5G vRAN function placement, lightpath provisioning, and network slicing. Section III presents the system model and problem formulation. Section IV provides the ILP formulation. Section V introduces heuristic and GA-based approaches. Section VI describes the experimental setup and reports the results. Section VII concludes the paper and outlines directions for future research.

\section{Related Work}
The problem of optimal functional placement in vRAN has been studied widely. Early efforts such as \cite{garcia2018wizhaul} and \cite{9389648} aim to increase centralization by reducing the computational cost of baseband processing. 
Other works, including \cite{sarikaya2021placement} and \cite{alabbasi2018optimal}, assign a centralization metric to the functional splits and select the split to optimize the overall level of centralization.

PlaceRAN \cite{morais2022placeran} formulates the placement problem as an Integer Linear Program that selects functional splits and places functions across CU and DU to maximize aggregation while minimizing computing resources under crosshaul and compute constraints. However, it does not model different slice types and their specific requirements. In particular, it does not capture availability, which is essential for URLLC, where downtime and added latency cannot be tolerated. It also does not include optical transport in its framework.
\begin{table*}[ht]
\centering
\caption{Comparison of the key features in the related work.}
\label{table 1}
\begin{threeparttable}
\begin{tabular}{lcccc}
\toprule
\textbf{Reference} & \textbf{Slice-awareness} & \textbf{Backup/Redundancy} & \textbf{Lightpath provisioning (Optical)} & \textbf{Split-aware function placement} \\
\midrule
Morais et al. \cite{morais2022placeran}          & \xmark      & \xmark                 & \xmark      & \checkmark \\
Da Silva et al. \cite{da2021function}            & \checkmark  & \xmark                 & \xmark      & \checkmark \\
Sen et al. \cite{sen2023slice}                   & \checkmark  & \xmark                 & \xmark      & \checkmark \\
Yusupov et al. \cite{yusupov2018multi}           & \checkmark  & \xmark                 & \xmark      & \checkmark \\
Yu et al. \cite{yu2020isolation}                 & \xmark      & \xmark                 & $\sim$\tnote{1} & $\sim$\tnote{2} \\
Xiao et al. \cite{xiao2020can}                   & \xmark      & \xmark                 & \checkmark  & \checkmark \\
Zorello et al. \cite{zorello2022power}           & \xmark      & \xmark                 & \xmark      & $\sim$\tnote{3} \\
Klinkowski et al. \cite{klinkowski2023optimized} & \xmark      & \xmark                 & \xmark      & \xmark \\
Marotta et al. \cite{9733024}       & $\sim$\tnote{4} & $\sim$\tnote{5}     & \checkmark  & \xmark \\
Sen et al. \cite{sen2025slice}                   & \checkmark  & \xmark                 & \xmark      & \checkmark \\
Li et al. \cite{li2024reliability}               & \xmark      & \checkmark             & \xmark      & \xmark \\
Begi\'c et al. \cite{begic2024framework}        & \checkmark  & \checkmark\tnote{6}    & \xmark      & \xmark\tnote{7} \\
Ahsan et al. \cite{ahsan2023efficient}             & \checkmark  & \xmark                 & \checkmark  & $\sim$\tnote{8} \\
\rowcolor{myrowgray}
\textbf{This work}                                & \checkmark  & \checkmark             & \checkmark  & \checkmark \\
\bottomrule
\end{tabular}
\begin{tablenotes}
\footnotesize
\item[1] Considers wavelength allocation but not full DWDM-aware lightpath provisioning.
\item[2] Acknowledges flexible RU/DU/CU splits and uses a fixed split in modeling; does not optimize split selection per slice/base station.
\item[3] Accounts for functional-split–specific latency/capacity but does not optimize split selection (DU/CU placement only).
\item[4] Slice isolation partially considered; not full per-slice logic.
\item[5] Backup applies only to lightpaths; node-level redundancy not addressed.
\item[6] Path redundancy (via 2-edge-connected subgraphs) handled; node redundancy is not.
\item[7] Generic slice/VNF embedding with path protection (Steiner trees); no RU/DU/CU split selection.
\item[8] Adopts 3GPP Option 7/7.2 with corresponding delay/bit-rate, but does not optimize split selection (policy-based processing across fog/BBU).
\end{tablenotes}
\end{threeparttable}
\end{table*}

Several studies consider slice awareness in functional splitting. Examples include \cite{da2021function}, \cite{sen2023slice}, and \cite{yusupov2018multi}. These works improve placement decisions by accounting for slice characteristics, but they do not include availability. In most cases, they also omit optical layer awareness and do not couple slice-aware decisions with joint function placement across CU/DU.

A complementary line of work examines resource and transport aspects. Yu et al. \cite{yu2020isolation} minimize the number of activated nodes and the number of allocated wavelengths using a node ranking heuristic with a focus on service function isolation. Xiao et al. \cite{xiao2020can} study a fine-grained split architecture and present a mixed integer linear program to optimize processing and data rate consumption. Zorello et al. \cite{zorello2022power} combine delay and compute constraints with split choices to reduce power consumption in RANs. Klinkowski et al. \cite{klinkowski2023optimized} consider flow routing and CU to DU planning in a vRAN over an xhaul packet switched network. They propose two heuristics and an MILP model to improve energy and resource use. These works do not include network slicing and availability. 
Marotta et al. \cite{9733024} address reliable slicing with dedicated lightpath protection for different degrees of VNF and network isolation in optical metro-aggregation networks. They formulate an MILP with a two-phase heuristic to jointly optimize VNF placement and working/backup lightpath provisioning. 
However, the functional split is fixed; per-haul latency constraints for diverse slice requirements are not considered, and protection is limited to link failures at the lightpath layer without node-level compute redundancy.
More recent work continues this trend. Sen et al. \cite{sen2025slice} propose an MILP that performs functional placement and processing based on each slice service requirement. However, the work does not include built-in redundancy or optical network considerations. Li et al. \cite{li2024reliability} focus on reliability through backup link provisioning in function chains. They use an MILP with explicit backup link allocation and present a column generation-based heuristic for larger networks. Their optimization is at the service function chain level rather than per slice, and it does not consider lightpath provisioning or optical constraints. 
Begic et al. \cite{begic2024framework} introduce a slice-level framework with survivability using 2-edge-connected subgraphs, formulating a revenue-maximizing ILP under latency and capacity constraints. Their approach provides link-failure survivability but does not address node redundancy, optical transport constraints, or functional split selection.
Ahsan et al. \cite{ahsan2023efficient} present an integrated framework for grooming, routing, and wavelength assignment in a three-layer WDM-enabled CF RAN. They consider slice awareness and transport layer optimization, but they do not add redundancy mechanisms for reliability.

These studies are promising. Most works, however, focus on one aspect such as slice awareness, reliability, or resource efficiency, or at most a pair of these aspects. To the best of our knowledge, no fully integrated approach exists (see Table \ref{table 1}). In many studies, transport constraints are often simplified, even though they are important for availability-aware RAN slicing over optical networks. As a result, availability-aware slicing, function placement, and optical routing have largely been treated in isolation, and many vRAN placement studies emphasize compute and latency while overlooking optical transport limits. Our work addresses this gap with an integrated framework that jointly considers availability, slice awareness, function placement, and optical transport. 


The framework proposed in this paper applies to multi-vendor, disaggregated NG-RAN deployments where operators must meet slice-level service agreements while using optical and compute resources efficiently. By embedding functional split awareness, slice-specific availability, and transport provisioning within one decision framework, it targets a practical need for operators deploying large 5G systems and preparing for 6G.
The standards landscape reinforces this need. The work in 3GPP Release 18 and Release 19 addresses RAN slicing and resiliency \cite{lin2025bridge}, while O-RAN Working Group 6 focuses on transport integration \cite{agarwal2025open}. 
ETSI NFV guidance on reliability complements these efforts. Collectively, these standards define how the radio and transport domains should work together and set clear expectations for slice-level service agreements. 

\begin{figure}
	\centering
	\includegraphics[width=9.5 cm, height=8 cm]{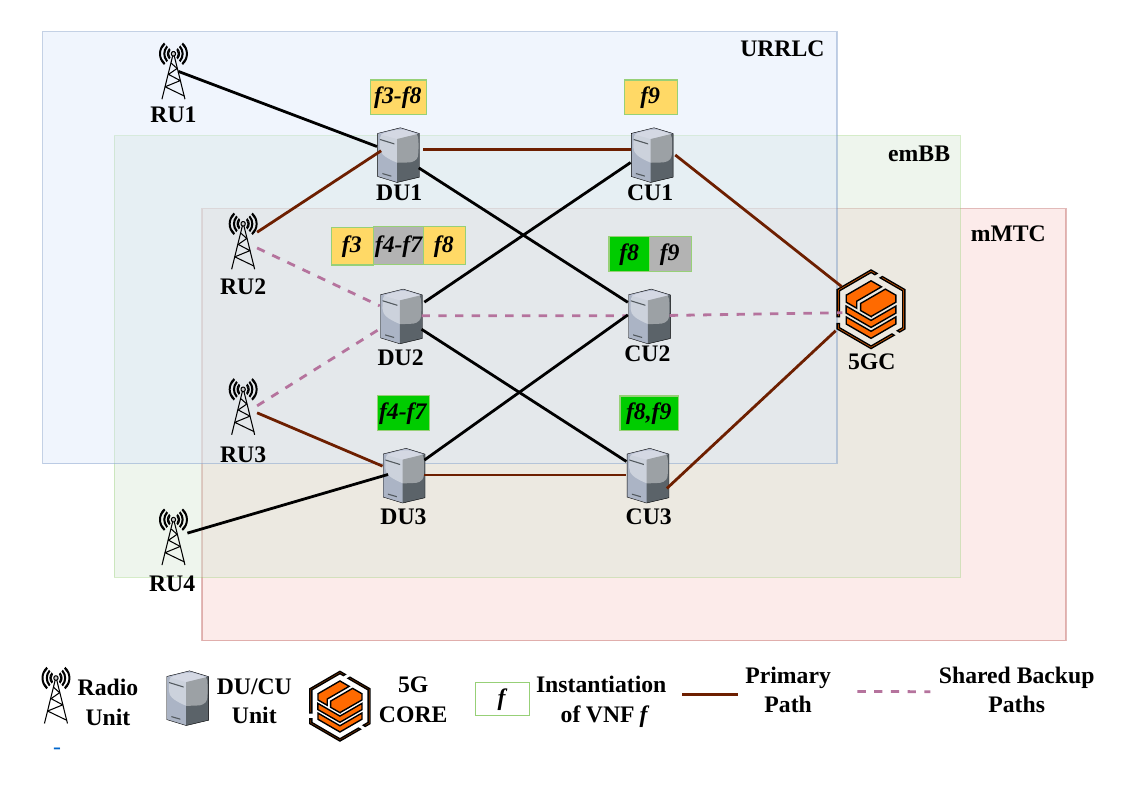}
	 \caption{Reference vRAN architecture and an illustrative example of the shared backup scheme.}
	\label{FIG:2}
 \vspace*{-5pt}
\end{figure}
\begin{figure*}[t]
	\centering
	\includegraphics[width=18 cm, height=10cm]{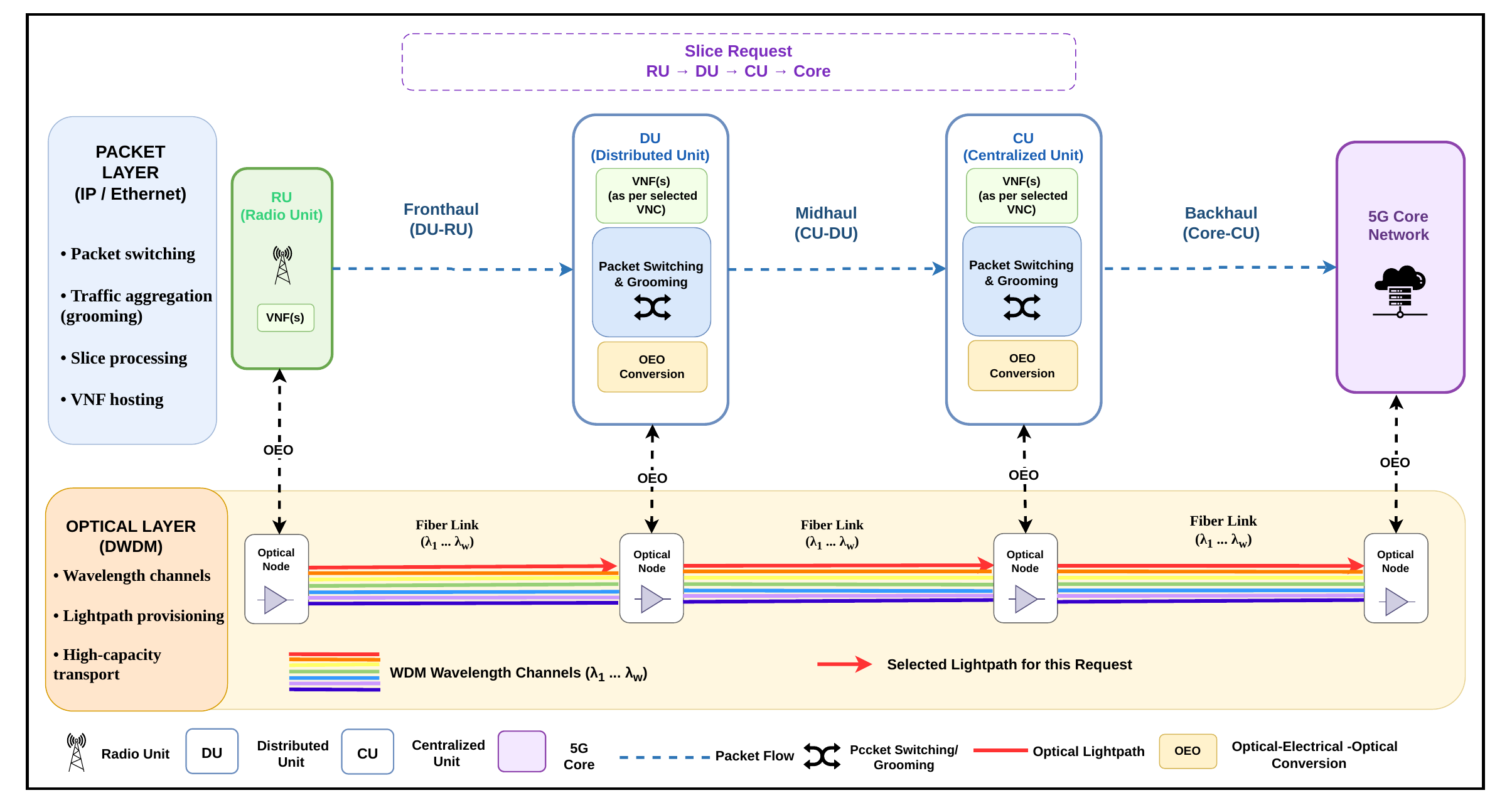}
	 \caption{Two-layer packet-over-optical transport architecture, comprising a DWDM optical layer providing wavelength-based lightpath connectivity and a packet layer at DU/CU nodes performing OEO conversion, switching, and grooming to carry aggregated slice traffic from RU to the 5G core.}
	\label{FIG:transport_network_packet_over_optical}
 \vspace*{-10pt}
\end{figure*}

An earlier version of this work appeared as SliAvailRAN \cite{ahmed2024sliavailran}. SliAvailRAN presented an ILP for vRAN that placed functions across CU, DU, and RU and admitted slice requests under availability, delay, and bandwidth limits. This paper extends SliAvailRAN in three main ways. First, we add optical transport to the model: we plan end-to-end lightpaths, choose primary and backup routes, and assign wavelengths while meeting latency and capacity budgets from RU to core. Second, we provide a formal analysis of the problem, including computational hardness and structural properties that guide algorithm design. Third, we introduce scalable solvers, a greedy heuristic, and a genetic algorithm, which allow us to handle larger networks. These additions align the framework with slice-level SLAs and the realities of NG-RAN optical transport. Also, this article augments the conference version with more extensive results and a broader review of related work.

\section{System Model}

\subsection{Preliminaries}

\subsubsection{Reference Architecture}
Fig. \ref{FIG:2} illustrates the reference vRAN architecture considered in this work. 
The vRAN supports three network slice types: URLLC, eMBB, and mMTC.
 User equipment (UEs) connect to the nearest RU, which is also connected to computing nodes (servers) that can operate as a DU, a CU, or both simultaneously for different slice requests, following an \textit{x-haul} architecture \cite{de2015xhaul,3gpp-38.401}. 
DU/CU nodes, in turn, connect to the 5G Core through a mesh topology. 

We assume a two-layer transport architecture as shown in Fig. \ref{FIG:transport_network_packet_over_optical}  consisting of an optical DWDM layer and a packet processing layer. 
The optical layer provides wavelength-based lightpath connectivity between RUs, DU/CU nodes, and the 5G core over fibre-based transport links. 
We assume ROADM-free CU/DU sites, where each CU/DU adds or drops an end-to-end wavelength using passive DWDM mux/demux components, without wavelength conversion in transit. Reconfigurable add/drop functionality and restoration mechanisms are confined to the 5G core.

At the DU/CU nodes, an L2 packet-processing fabric (e.g., compact Ethernet switch or server NIC/DPU) terminates optical signals through OEO conversion and performs packet-level operations such as routing, per-slice QoS enforcement, packet switching, and traffic grooming. 
Consequently, multiple slice flows may be multiplexed onto the same wavelength channel before transmission. 
Accordingly, for each slice request, the path from the serving RU to the core node is provisioned as an optical lightpath, while packet-level forwarding and aggregation are handled at the DU/CU processing nodes. 
This design reflects practical cost- and latency-efficient x-haul deployments, where DU/CU sites primarily perform processing and traffic grooming while the optical layer remains transparent with simplified 1+1 or 1:1 protection provisioning.

\subsubsection{Virtual Network Configuration}
\label{section:vnc}

3GPP NG-RAN defines several functional split options that specify how the radio protocol stack (PHY, MAC, RLC, PDCP, RRC) is divided among the CU, DU, and RU \cite{3gpp-38.801}. These splits allow operators to trade off latency, bandwidth, and pooling efficiency by deciding whether functions should remain close to the RU or be shifted upward into the DU or CU. 
Modern NG-RAN and vRAN architectures enable flexible deployment organizations with different degrees of functional centralization and distributed processing depending on deployment scenario, virtualization capability, fronthaul characteristics, and service requirements. 
Although current commercial O-RAN deployments predominantly adopt Split 7.2x between RU and DU, together with Split 2 between DU and CU, different deployment scenarios may favor different functional split organizations.
In particular, Split 7.x is better suited for dense urban and high-capacity deployments, whereas Split 6/Small Cell Forum (SCF)-style and integrated D-RAN-like architectures are often preferred in transport-constrained, small-cell, and edge-centric deployment scenarios \cite{scf_splits_blog_2023}.

To capture the deployment-level heterogeneity enabled by modern NG-RAN and vRAN architectures, it is necessary to model different feasible functional deployment organizations that may coexist across infrastructure environments and service scenarios. 
Accordingly, in our earlier work \cite{ahmed2024sliavailran}, we introduced the concept of Virtual Network Configuration (VNC), defined as a specific mapping of radio VNFs onto a given arrangement of CUs, DUs, and RUs derived from standardized functional split definitions. 
In this study, we consider nine practically relevant VNCs representing different degrees of functional centralization and distributed processing, inspired by deployment organizations such as D-RAN, NG-RAN-II, NG-RAN-III, and C-RAN \cite{3gpp_ts23501_2018,gstr2018transport,3gpp-38.801,scf_splits_blog_2023}. 
The selected VNCs are restricted to practically adopted split organizations involving higher-layer options such as Split  1/2 (O1/O2) and lower-layer options such as Split 6/7 (O6/O7), and are chosen based on industry standards and practical deployment insights to ensure real-world applicability \cite{morais2022placeran}.

NG-RAN-III, consisting of VNCs 9, 8, 7, and 6, distributes functions across three independent nodes (CU, DU, and RU) and corresponds to 3GPP split options most widely deployed in practice. 
C-RAN, represented by VNCs 5 and 4, collapses the CU and DU into a single node, leaving the RU as a simple radio head and thereby concentrating most functions centrally. 
NG-RAN-II, which includes VNCs 3 and 2, collapses the DU and RU into one node while keeping the CU separate. Finally, VNC1 represents the traditional D-RAN, in which CU, DU, and RU are fully co-located at a single site.

In our framework, we explicitly order the VNCs such that a higher VNC number corresponds to a higher priority. Thus, VNC 1 has the lowest priority and VNC 9 the highest. This ordering does not simply follow the degree of centralization, but rather reflects practical deployment considerations. 
NG-RAN-III configurations (VNCs 6–9) are assigned the highest priority because they represent the most practically adopted disaggregated deployment organizations in modern NG-RAN, O-RAN, and SCF systems, offering an effective balance between functional centralization, transport feasibility, and resource pooling gains. 
 
C-RAN configurations (VNCs 4–5) are ranked below NG-RAN-III because highly centralized deployments impose stringent fronthaul latency and bandwidth requirements, which significantly constrain transport feasibility \cite{gstr2018transport}. 
NG-RAN-II configurations (VNCs 2–3) represent partially integrated DU/RU deployments that reduce transport demands at the expense of reduced pooling flexibility. Finally, D-RAN (VNC 1) is assigned the lowest priority because fully distributed deployments provide minimal centralization benefit and limited resource aggregation capability \cite{gstr2018transport}.

Within the NG-RAN-III, C-RAN, and NG-RAN-II VNC groups, O7-based configurations are prioritized over O6 because O7 enables stronger PHY-layer centralization, improved sharing of processing resources, and better alignment with current O-RAN deployment practices. 
Similarly, configurations based on O2 are prioritized over O1 because O2 enables greater centralization and resource pooling while still maintaining practical transport requirements. 
This priority structure is therefore not arbitrary, but reflects deployment tendencies and transport feasibility considerations, where the key trade-off lies between centralization, transport requirements, and deployment feasibility. 


\begin{figure}[t]
	\centering
	\includegraphics[width=9 cm, height=6.5 cm]{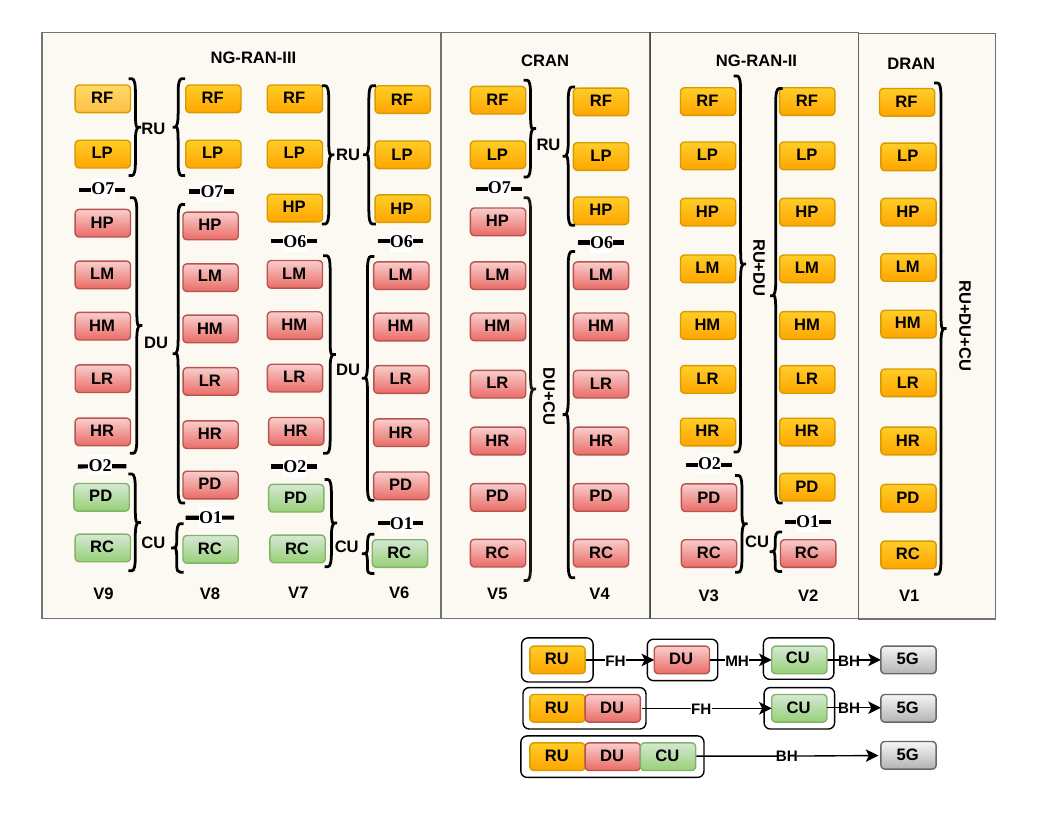}
	 \caption{Different configurations for placements of radio network functions on RU, DU, and CU nodes, defined by VNCs.}
	\label{FIG:vnc}
 \vspace*{-5pt}
\end{figure}




\subsection{Key Design Considerations for High-Availability vRAN}

\subsubsection{CU/DU Disjoint Backup Paths}
\label{availability}

Each slice request, once mapped to the vRAN, is associated with a unique VNC.
We assume that the availability requirement for each slice request is given. 
To satisfy that, we need to place VNFs of the associated VNC in redundant CU/DU nodes on the path traversed by the slice request. 
Also, we would need backup paths such that redundant CUs and DUs are node-disjoint (parallel formation). This is because if any CU/DU node in a path fails, the other redundant node disjoint path takes over. 
Two paths are considered to be CU/DU disjoint if no common CU/DU node exists in both paths. 
Consider an example vRAN shown in Fig. \ref{FIG:2} with a core node (5GC) and multiple CU/DU nodes. 
In this example, paths $RU2-DU1-CU1-5GC$ and $RU2-DU2-CU1-5GC$ are a pair of DU disjoint paths. While paths $RU2-DU1-CU1-5GC$ and $RU2-DU1-CU2-5GC$ is a pair of CU disjoint paths. 
Whereas, $RU2-DU1-CU1-5GC$ (in brown) and $RU2-DU2-CU2-5GC$ (dashed line) are a pair of both DU and CU disjoint paths. 
For a given slice request, its primary and backup path(s) must be both DU and CU disjoint.

\par We can calculate the required number of CU and DU disjoint paths to achieve an availability target for a slice request using the following expression. \\
\begin{equation} \label{eq:avail-formula}
   \tau = \left\lceil\frac{log(1-\sqrt{\Gamma_{s}})}{log(1-\chi_{n})}\right\rceil, 
\end{equation}

where $\Gamma_{s}$ denotes the availability requirement of slice request $s$ at an RU $r$, $\chi_{n}$ is the availability of a node $n$ (where CU/DU is instantiated), and $\tau$ is the number of CU/DU disjoint paths required to satisfy the availability requirement of the slice request (including the primary path). 
For example, assuming that the node availability ($\chi_{n}$) is 0.999, and the availability requirement ($\Gamma_{s}$) of slice request $s$ is 0.999, then from eq. \ref{eq:avail-formula}, we get $\tau = 2$.  
This denotes that we need two CU/DU disjoint paths to satisfy the availability requirement.

\subsubsection{Shared Backup} \label{subsec:shared} 

The shared backup scheme allows multiple slice requests to reuse the same backup node, thereby optimizing network resources. In this approach, if two or more slice requests are allocated backup paths that share a backup node, and their selected VNCs place the same set of VNFs at the same level (DU or CU), then only a single instantiation of the backup VNF is required. This shared instance serves all such requests, avoiding redundant instantiations and improving efficiency.

Consider Fig.~\ref{FIG:2} to better understand this concept. Assume a slice request A originates at RU2 and a slice request B originates at RU3. The primary lightpaths selected are $RU2–DU1–CU1–5GC$ for request A and $RU3–DU3–CU3–5GC$ for request B (solid brown lines). The corresponding backup paths are $RU2–DU2–CU2–5GC$ for A and $RU3–DU2–CU2–5GC$ for B (dashed lines), thereby sharing DU2 and CU2 as common backup nodes.

Assume request A employs VNC $8$, where functions $f3$–$f8$ are placed at the DU and function $f9$ at the CU (shown in yellow), while request B employs VNC $7$, where functions $f4$–$f7$ are placed at the DU and functions $f8$–$f9$ at the CU (shown in green). 
Since functions $f4$–$f7$ overlap at the DU level and function $f9$ overlaps at the CU level, these VNFs need only be instantiated once at the shared backup nodes DU2 and CU2. 
The shared VNF instances (shown in grey) can simultaneously serve both slice requests. This reduces redundant backup instantiations and lowers the overall compute resource consumption, which is the key benefit of the shared backup scheme.

\section{ILP Formulation}

We model the vRAN as a graph $G(N, E)$, where a subset of nodes represent RUs, CU/DU and a single core node (CN) represents the 5G core. Each RU may connect to the CN through multiple possible paths, as illustrated in Fig. 2. Slice requests originate at an RU and must be mapped onto one of these RU–CN paths. If a request is admitted, a VNC is selected and the required VNFs are instantiated at the appropriate CU, DU, or RU nodes, incurring costs for the mobile network operators (MNOs). At the same time, accepting slice requests generates revenue depending on the slice type.
In this section, we formulate an Integer Linear Programming (ILP) model to maximize the MNO’s total profit by balancing placement costs and slice revenues. We define this optimization problem as \textit{Slicing-Aware vRAN Functional Placement with Lightpath Provisioning (SA-vFPLP)}.
Tables \ref{tab:notations}, \ref{tab:binary_input} and \ref{tab:binary_decision} define the notations, binary input parameters and binary decision variables of our ILP formulation, respectively.
\begin{table}[h]
\centering
\caption{Notations}
\begin{tabular}{|c|l|}
\hline
\textbf{Notation} & \textbf{Description} \\ \hline
$G(N,L)$          & Graph with $N$ nodes and $L$ vertices                     \\ \hline
$N$              & Set of processing nodes                               \\ \hline
$L$               & Set of links $l \in L$                               \\ \hline
$\Upsilon_n$          & Processing capacity of node $n \in N$                    \\ \hline
$\xi_{l}$ & Bandwidth capacity of link $l$ \\ \hline
$\chi_{n}$           & Availability of node $n \in N$                           \\  
\hline
$\Gamma_{s}$& Availability requirement of a slice $s$ \\ \hline
$\tau$ &  Number of node disjoint (CU/DU) paths
needed to satisfy \\ & the availability requirement $\Gamma_{s}$\\ \hline
$R$               & Set of RU nodes, $R \subset N$ and $r \in R$
\\ \hline
$W$ & Set of wavelengths ($\lambda \in W$).\\ \hline
$S$               & Set of slices, $s \in S$, $S=\{ S_{0}, S_{1}, S_{2}\}$ \\ &  where $S_{0}=$ URLLC, $S_{1}=$ mMTC and $S_{2}=$ eMBB                                                                                                           \\ \hline
$P_{r}$       & Set of paths originating from RU $r$ to CN                                                     \\ 
\hline
$P$       & Set of all paths, such that: $\sum_{r \in R} P_r = P$.                                                     \\ 
\hline 
$H$                & Set of types of path segments, where $h \in H$ \text{ and }\\ & $H=\text{\{0:Backhaul,  1:Midhaul, 2:Fronthaul\}}$                                                                               \\ \hline
$L_{h}^p$   &    Delay incurred at haul $h$ on path $p$.                                \\ \hline
$V$                & Set of VNCs, where $v \in V$ and $V = \{1,2,3,4,5,6,7,8, 9\}$                                                                                \\ \hline
$F$                & Set of all radio network functions and $f \in F$, $F$ =  ($f1 \rightarrow$ \\ & RF,  $f2 \rightarrow$ LP, $f3 \rightarrow$ HP  $f4 \rightarrow$ LM $f5 \rightarrow$ HM, $f6 \rightarrow$ LR, \\ & $f7 \rightarrow $ HR,  $f8 \rightarrow$ PD, $f9 \rightarrow$ RC)    \\  \hline 
$U$                 & Set of node types (levels).  $U=  \{0:CU, 1: DU, 2:RU\}$ \\ & and $u \in U$ \\ \hline
$CPU_{u}^{v}$      & Processing requirement of VNC $v$ \text{at level} $u \in \text\{0, 1\}$\\
\hline
$\sigma_{h}^{vs}$       & Latency requirement for a given (valid) combination of \\ & VNC $v$  and slice type $s$ at haul $h$ \\ \hline
$b^{vs}_h$ & Bandwidth requirement for a given (valid) combination of \\ & VNC $v$  and slice type $s$ at haul $h$ \\ \hline
$C_{act}$ & Cost of activating a single node\\ \hline
$C_{\lambda}^{l}$ & Cost of activating wavelength \( \lambda \) on link \( l \).\\ \hline
$Cost_{f}^{u}$ & Cost of instantiation of function $f$ at level $u$\\ \hline
$Rev^{s}$ & Revenue from accepting a request of slice type $s\in S$\\ \hline
$B_\lambda$ & Capacity of wavelength $\lambda$ \\ \hline

\end{tabular}
\label{tab:notations}
\end{table}
\begin{table}[h]
\centering
\caption{Binary Input Parameters}
\begin{tabular}{|c|l|}
\hline
\textbf{Variable} & \textbf{Meaning} \\ \hline
    $y^s_r$ & 1, if a slice of type $s$ is assigned at an RU $r$, \\ & such that $\sum_{s \in S} y^{s}_{r} = 1 \hspace{0.1 cm} \forall r \in R$\\
    \hline
    $\psi_{n}^{pu}$ &  1, If node $n$ belongs to path $p \in P_{r}^{k}$ for level $u$.\\
      \hline
      $x_{l}^{ph}$ &  1, If the given link $l$ is a part of a haul $h$ of path $p$. \\
      \hline
      $A_{n}$ &  1, If at least one function $f$ is deployed on the node $n$.\\
      \hline
      $W^{v}_{uf}$ & 1, if function $f$ is a part of VNC $v$ at level $u \in U$. \\
      \hline
       $\eta^v_u$ & 1, if at least one function is to be hosted at level \\ & $u \in U$ according to VNC $v$. \\
       \hline
\end{tabular}
\label{tab:binary_input}
\end{table}
\begin{table}[h]
\centering
\caption{Binary Decision Variables}
\begin{tabular}{|c|l|}
\hline
\textbf{Variable} & \textbf{Meaning} \\ \hline
$\alpha_{r}^{p}$ & 1 if path $p$  is selected as the primary path to service \\ & a slice at RU $r$. \\ \hline
$\beta^{p}_{r}$ & 1 if path $p$  is selected as a backup path to service \\ & a slice at RU $r$.  \\ \hline
$\mu_p^{\lambda}$ & 1, if a path $p$ is assigned a wavelength $\lambda$ \\ \hline
$Z^{v}_{rs}$ & 1, if VNC $v$ has been allocated to an RU $r$. \\ \hline
$ \gamma^{nu}_f$ & 1, if function $f$ has to run to a backup node \\ & $n$ at level $u$ (DU/CU)\\ \hline
$\Phi_r^s$ & 1, if service request of slice type $s$ is accepted at RU $r$ \\ \hline
\end{tabular}
\label{tab:binary_decision}
\end{table}
\subsection{Costs and Revenue}

\subsubsection{Cost of node activation}
A node $n$ is considered activated when at least one function is instantiated on it, and this cost captures the fixed overhead of bringing that node into service.
The total cost of node activation is given by:
\begin{equation}
C_{N_a} = \sum_{n \in N} A_n \cdot C_{act}
\end{equation}

\subsubsection{Cost of VNF instantiation on a node}
The cost for instantiating VNFs differs for RU, DU, and CU nodes. The total cost of VNF instantiation is given by: 
\\(i) For shared backup scheme:
\begin{equation}
\begin{split}
    C_{VNF_i} = \sum_{n \in N} \sum_{f \in F} \sum_{r \in R} \sum_{s \in S} \sum_{v \in V} \sum_{p \in P_r} \bigl( \psi^{pu}_{n} \cdot\alpha_{rs}^{p} ( W^{v}_{uf} \cdot \text{Cost}_f^u )\bigr)\\+ \sum_{n \in N} \sum_{f \in F}
    \sum_{u \in U} \gamma^{nu}_f \cdot \text{Cost}_f^u
\end{split}
\end{equation}
(ii) For unshared backup scheme:\\
\begin{equation}
\begin{split}
    C_{VNF_i} =\sum_{n \in N} \sum_{f \in F} \sum_{r \in R} \sum_{s \in S} \sum_{v \in V} \sum_{p \in P_r} \bigl( \psi^{pu}_{n} (\alpha_{rs}^{p} + {\beta}_{rs}^{p}) \\\times ( W^{v}_{uf} \cdot \text{Cost}_f^u ) \bigr)
\end{split}
\end{equation}
In the above equation, binary variable \( W^{v}_{uf} \) ensures that only relevant costs are considered (i.e., only if the function $f$ is the part of a given VNC $v$ at level $u$). 

\subsubsection{Cost of wavelength activation}
The cost of activating a wavelength represents the expense of provisioning an optical channel in the DWDM layer, which includes transponder usage, spectrum allocation, and associated optical equipment costs required to establish a lightpath.
The cost of wavelength activation for each link is given by:
\begin{equation}
C_{W_a} = \sum_{\lambda \in W} \sum_{l \in L} \sum_{h \in H} \sum_{p \in P} x_{l}^{ph} \cdot \mu^{\lambda}_{p} \cdot C^{l}_{\lambda}
\end{equation}

\subsubsection{Revenue from accepted requests}
The revenue generated by admitting different slice types varies. The total revenue from all accepted requests is expressed as
\begin{equation}
R_{act} = \sum_{r \in R} \sum_{s \in S} \Phi^{s}_{r} \cdot Rev^{s}
\end{equation}

\subsection{Objective function}
The objective function maximizes the net profit, considering the revenues and costs defined above:
\begin{equation}
\begin{split}
   \text{Maximize} (R_{act} - C_{N_a} - C_{VNF_i} - C_{W_a}) 
\end{split}
\end{equation}

\subsection{Constraints}
We now define the constraints that the solution must satisfy.
\subsubsection{1-to-1 mapping of the RU and VNC}
Each accepted service request of a given slice type is mapped to exactly one VNC.
\begin{equation}
    \sum_{v \in V} Z^{v}_{rs}= \Phi^s_r \hspace{0.5 cm} \forall s \in S , \forall r \in R
\end{equation}


\subsubsection{Primary path constraint}
For each RU $r$ that has been allocated a particular service request of slice type $s$, exactly one primary path must be selected.
\begin{equation}
\sum_{p \in P_r} \alpha_{rs}^{p} = \Phi^s_r \hspace{0.5cm} \forall s \in S, \forall r \in R
\end{equation}
\subsubsection{Backup path constraint}
Each slice request that has been allocated at an RU $r$ (of slice type $s$), must be allocated $\tau -1$ backup paths for satisfying the reliability requirements of that slice type. Here, $\tau$ is calculated using a formula based on the availability requirement of a slice type and the availability of a single node $n$ in the network, as described in subsection \ref{availability}.
\begin{equation}
    \sum_{p \in P_{r}} \beta^{p}_{rs} = \left( \tau -1 \right) \cdot \Phi^s_r \hspace{0.5 cm}  \forall s \in S, \forall r \in R
\end{equation}
The product in the RHS restricts the number of backup paths to be exactly $\tau -1$.
\subsubsection{Backup and primary path disjointedness assurance}
This constraint ensures that for a given RU and slice allocated to that RU, a path $p$ can either be chosen as a primary path, or as a backup path, but NOT both together simultaneously. This is given by:
\begin{equation}
{\alpha}_{rs}^{p} + {\beta}_{rs}^{p} 
\leq 1 \hspace{0.5cm} 
\forall p \in P_r,
\forall s \in S, \forall r \in R.
\end{equation}

\subsubsection{Mapping between path and wavelength}
For any RU $r$, the selected primary and backup path should be assigned a unique wavelength. 
\begin{equation}
    \sum_{\lambda \in W} \mu^{\lambda}_p \leq 1 \hspace{0.3cm} \forall p \in P_r, \forall r \in R
\end{equation}
The above equation states that each path should be assigned at most one wavelength channel.
\begin{equation}
    \alpha^p_{rs} + \beta^p_{rs} \leq \sum_{\lambda \in W} \mu^{\lambda}_p \hspace{0.3cm} \forall p \in P_r, \forall s \in S, \forall r \in R
\end{equation}
This inequality guarantees that if a path $p$ serves as either a primary or backup path for an RU $r$ (when selected), it is assigned a wavelength.
\subsubsection{CU/DU node disjointedness}
\label{disjointness}
The primary and backup paths selected for a specific slice type $s$ allocated to an RU $r$ should have DU/CU node disjointness to fulfil the reliability requirements.

\begin{equation}
\begin{split}
       \alpha^{p}_{rs} \cdot \psi_{n}^{pu} +\beta^{p'}_{rs} \cdot \psi_{n}^{p'u} \leq 1 \hspace{0.5 cm} \forall n \in N, \forall p,p' \in P_r \hspace{0.1 cm} s.t. \hspace{0.1 cm} p \neq p'\\, \forall u \in \{0,1\},\forall s \in S, \forall r \in R
\end{split}
\end{equation}
\begin{equation}
\begin{split}
       \beta^{p}_{rs} \cdot \psi_{n}^{pu} +\beta^{p'}_{rs} \cdot \psi_{n}^{p'u} \leq 1 \hspace{0.5 cm} \forall n \in N, \forall p,p' \in P_r \hspace{0.1 cm} s.t. \hspace{0.1 cm} p \neq p'\\, \forall u \in \{0,1\},\forall s \in S, \forall r \in R
\end{split}
\end{equation}
Eq. 14 ensures that the primary path and the backup path selected for the same request will not share the same node $n$ acting as the DU or CU for both the paths. Eq. 15 ensures that if there is more than one backup path to be assigned (based on the value of $\tau$ for the given slice request), then these backup paths should also be DU/CU disjoint with each other. 

\subsubsection{Shared backup constraint}: 
Consider a scenario where backup paths for different slice requests pass through common nodes (CU/DU) and are assigned the same VNC. If these requests require the same function $f \in F$ at the same level $u \in U$ on those nodes, we instantiate a single VNF that is shared across all such requests, rather than deploying separate VNFs for each.
\begin{equation}
\begin{split}
Z^v_{rs} \cdot \beta_{rs}^p \cdot \psi_{n}^{p u} \cdot W_{uf}^v \leq \gamma_{f}^{nu} \hspace{0.2cm} \forall n \in N, \forall s \in S, \forall r \in R, \forall v \in V, \\ \forall f \in F, \forall u \in U, \forall p \in P_r
\end{split}
\end{equation}
The RHS limits the number of instances of function $f$ running at the DU/CU level of node $n$ (under VNC $v$) on the backup path of a request $r$ to at most one. Consequently, when multiple requests traverse the same node requiring the same function, they activate the same $\gamma_{f}^{nu}$, ensuring a single instantiation that is shared across them.

\subsubsection{Processing capacity constraint}
This constraint ensures that the total processing load on each node $n$, arising from running the VNFs of different slice requests associated with it, does not exceed the node’s processing capacity.
\begin{equation}
\begin{split}
\sum_{r \in R} \sum_{s \in S} \sum_{v \in V} \sum_{p \in P_r} \sum_{u \in U} \psi_{n}^{pu} \cdot (\alpha^p_{rs} + \beta^p_{rs}) \cdot Z^v_{rs}\cdot CPU^{v}_{u} & \leq \Upsilon_n \hspace{0.5 cm} \\\forall n \in N 
\end{split}
\end{equation}

\subsubsection{Latency constraint}
The latency of any link which is acting as a haul $h$ for a path $p$ should not exceed the maximum tolerable latency as per the VNC $v$ and slice $s$ assigned to that path, $ \sigma^{vs}_h$. 
\begin{equation}
\begin{split}
    x_l^{ph} \cdot (\alpha^p_{rs} + \beta^p_{rs}) \cdot Z^v_{rs} \cdot L^p_h \leq \sigma^{vs}_h \hspace{0.5 cm} 
    \\ \forall p \in P, \forall h \in H, \forall v \in V, \forall l \in L,\forall s \in S,\forall r \in R,
    \end{split}
\end{equation}
where $\sigma^{vs}_h$ is determined as the \textit{minimum} of the maximum tolerable latency for the given haul $h$ for the chosen VNC and the chosen slice type respectively. 


\subsubsection{Link capacity constraints}
This constraint ensures that the total bandwidth consumed on all paths traversing a link does not exceed the link's available capacity.
\begin{equation}
      \sum_{r \in R}\sum_{s \in S} \sum_{p \in P_r} \sum_{v \in V} (\alpha^p_{rs} + \beta^p_{rs}) \cdot Z^v_{rs} \cdot x^{ph}_{l} \cdot b^{vs}_h \leq \xi_{l} \hspace{0.5 cm} \forall l \in L,
\end{equation}
where $b^{vs}_h$ is computed as the \textit{maximum} of the bandwidth requirements corresponding to the chosen VNC for the haul and the request’s specific slice type.

Under WDM, each wavelength $\lambda \in W$ operates as an independent data channel. The combined capacity of these channels on a given link must not exceed the link capacity.

\begin{equation}
    \sum_{p \in P} \sum_{\lambda \in W} \sum_{h \in H} x^{ph}_l \cdot \mu^p_{\lambda} \cdot B_{\lambda} \leq \xi_{l}  \hspace{0.3cm}\forall l \in L
\end{equation}

\subsubsection{Unique wavelength assignment constraint}
Ensures unique wavelength assignment on any given link:
\begin{equation}
\sum_{h \in H} \sum_{p \in P_r} x^{ph}_{l} \cdot \mu^{\lambda}_{p} \leq 1\hspace{0.5 cm} \forall \lambda \in W, \forall l \in L.
\end{equation}
This constraint ensures that a wavelength \( \lambda \) is assigned to at most one path passing through link \( l \). 


\subsubsection{Activation constraint}
This constraint ensures that if any VNF $f$ is hosted on any node $n$, then the node should be marked activated. 
\begin{equation}
\begin{split}
    \sum_{p \in P_r} \sum_{u \in U} \sum_{v \in V} \left( \alpha_{rs}^p + \beta^p_{rs} \right) \cdot Z^{v}_{rs} \cdot \psi^{pu}_n \cdot \eta_{u}^{v}\leq A_n \hspace{0.5cm}\\ \forall s \in S, \forall r \in R, \forall n \in N.
\end{split}
\end{equation}
If any VNF instance of type $f$ is placed at node $n$ for a given request, the node must be activated (i.e., $A_n=1$), thus linking function placement decisions on the LHS with node activation on the RHS.


\subsubsection{Wavelength capacity constraint}


Each slice request is associated with a certain bandwidth requirement ($b^{vs}_h$). When multiple requests are mapped onto the same wavelength channel $\lambda$, the total bandwidth allocated to that wavelength must not exceed its available capacity $B_\lambda$. This ensures that the cumulative demand of all requests using wavelength $\lambda$ remains within the physical transmission limit of that channel.

\begin{equation}
    \sum_{r \in R} \sum_{s \in S} \sum_{v \in V} \sum_{p \in P} \sum_{h \in H} (\alpha_{rs}^p + \beta^p_{rs}) \cdot Z^{v}_{rs} \cdot \mu^{\lambda}_p \cdot b^{vs}_h \leq B_{\lambda} \hspace{0.3cm} \forall \lambda \in W
\end{equation}

\subsubsection{Linearization of non-linear constraints}
In the above formulation, some constraints are non-linear since they involve the product of two binary decision variables. We address these by introducing auxiliary variables, handling them one at a time.
\par The processing/link capacity constraints, latency constraint and activation constraint have non-linearity because of the term: $\left( \alpha_{rs}^p + \beta^p_{rs} \right) \cdot Z^{v}_{rs}$, so we introduce the term $\hat{\alpha}^{pv}_{rs}$ and $\hat{\beta}^{pv}_{rs}$, where $\hat{\alpha}^{pv}_{rs}= \alpha_{rs}^p \cdot Z^{v}_{rs}$ and $\hat{\beta}^{pv}_{rs} = \beta^p_{rs} \cdot Z^{v}_{rs}$,  such that:
\begin{equation}
    \hat{\alpha}^{pv}_{rs} \leq Z^{v}_{rs}
\end{equation}
\begin{equation}
    \hat{\alpha}^{pv}_{rs} \leq \alpha_{rs}^{p}
\end{equation}
\begin{equation}
    \hat{\alpha}^{pv}_{rs} \geq Z^{v}_{rs} + \alpha_{rs}^{p} - 1
\end{equation}
\begin{equation}
    \hat{\beta}^{pv}_{rs} \leq Z^{v}_{rs}
\end{equation}
\begin{equation}
   \hat{\beta}^{pv}_{rs} \leq \beta_{rs}^{p}
\end{equation}
\begin{equation}
    \hat{\beta}^{pv}_{rs} \geq Z^{v}_{rs} + \beta_{rs}^{p} - 1
\end{equation}
Then, we can resolve the non-linearity for all these constraints by rewriting the term $(\alpha^p_{rs} + \beta^p_{rs}) \cdot Z^v_{rs}$ as:
\begin{equation*}
    (\alpha^p_{rs} + \beta^p_{rs}) \cdot Z^v_{rs} = (\alpha^p_{rs} \cdot Z^v_{rs} + \beta^p_{rs} \cdot Z^v_{rs})
\end{equation*}
\begin{equation*}
   = ( \hat{\alpha}^{pv}_{rs} + \hat{\beta}^{pv}_{rs})
\end{equation*}
Note that the shared backup constraint is also non-linear due to the product of $Z^v_{rs}$ and $\beta^p_{rs}$, which can be linearized by using $\hat{\beta}^{pv}_{rs}$ defined above.
\par Finally, the wavelength capacity constraint is also non-linear, because of the product of two decision binaries, even if we use 
$(\hat{\alpha}^{pv}_{rs}+ \hat{\beta}^{pv}_{rs})$ (because this term would still be multiplied by $\mu^{\lambda}_p$). Thus we introduce two auxiliary variables: $\tilde{\alpha}^{pv\lambda}_{rs}$  and $\tilde{\beta}^{pv\lambda}_{rs}$ where $\tilde{\alpha}^{pv\lambda}_{rs} = \hat{\alpha}^{pv}_{rs} \cdot \mu^\lambda_p$ and $\tilde{\beta}^{pv\lambda}_{rs}= \hat{\beta}^{pv}_{rs} \cdot \mu^\lambda_p$ such that:
\begin{equation}
    \tilde{\alpha}^{pv\lambda}_{rs} \le \hat{\alpha}^{pv}_{rs}
\end{equation}
\begin{equation}
    \tilde{\alpha}^{pv\lambda}_{rs} \le \mu^\lambda_p
\end{equation}
\begin{equation}
    \tilde{\alpha}^{pv\lambda}_{rs} \ge \hat{\alpha}^{pv}_{rs} + \mu^\lambda_p - 1
\end{equation}
and similarly for $\tilde{\beta}^{pv\lambda}_{rs}$.

\section{Heuristic and Metaheuristic approaches}


We now establish the computational complexity of \textit{Slicing-Aware vRAN Functional Placement with Lightpath Provisioning (SA-vFPLP) problem}. 

\begin{theorem}
	SA-vFPLP is NP-hard.
\end{theorem}

\begin{proof}
	We reduce from the Multi-dimensional Multiple-choice Knapsack Problem (MMKP) \cite{akbar2006solving}. MMKP consists of $M$ groups of items and $K$ resource types. Each item $j$ in group $i \in M$ has a value $v_{i,j}$ and consumes $z_{i,j,k}$ units of resource $k \in K$. The goal is to select exactly one item from each group to maximize the total value without exceeding the capacity $Z_k$ of any resource.
	
	We define a restricted version of our problem as follows. Each RU corresponds to a group in MMKP, and each valid tuple of \emph{(VNC, primary path, backup path set, wavelength)} corresponds to an item. 
Node processing capacity and link bandwidth map to the two resource types ($K=2$), and the profit (revenue minus cost) from selecting a tuple is the value of that item. The constraint that exactly one tuple must be chosen for each RU matches the one-item-per-group rule in MMKP. Thus, maximizing profit in our problem is equivalent to maximizing value in MMKP.
	
	This reduction is polynomial in input size. Since MMKP is NP-hard, our restricted problem is also NP-hard, i.e., MMKP $\leq_P$ SA-vFPLP. The full formulation further includes node-disjoint and shared backups, making it at least as hard as the restricted case. Therefore, the general vRAN functional placement with wavelength provisioning problem is NP-hard.
	
	This implies that while the ILP provides exact solutions for small networks, it does not scale to large topologies, which motivates the use of heuristic and metaheuristic approaches.
\end{proof}

\subsection{Greedy Heuristic Approach}

Due to the computational complexity of the ILP approach, we also propose a greedy heuristic for the vRAN functional placement problem with lightpath provisioning. The algorithm proceeds sequentially, beginning by sorting RUs according to slice type priority, with URLLC ranked highest, followed by eMBB and mMTC. 
For each RU, the required number of backup paths is determined from the availability parameter $\tau$. The $K$ shortest paths between every RU and the CN are precomputed.  Among these, the algorithm selects one as the primary path and then chooses the $\tau-1$ node-disjoint backup paths to satisfy availability requirements.

The algorithm then iterates through VNCs in descending order of priority. For each candidate primary path, it checks feasibility against processing, latency, and bandwidth constraints. Once a feasible primary path is identified, wavelengths are assigned to both primary and backup paths while ensuring wavelength uniqueness and capacity constraints. After a VNC, wavelength, and primary/backup paths are allocated for a request, the algorithm updates the resource state, including node processing capacity, link bandwidth, and wavelength usage, before moving to the next RU.

This heuristic strikes a balance between computational efficiency and solution quality. It achieves significantly shorter execution times, even for large topologies, making it scalable. However, as with most greedy algorithms, it does not guarantee globally optimal solutions. This limitation becomes more evident in highly constrained scenarios with heavy resource contention, where the heuristic may yield suboptimal placements compared to exact ILP-based solutions.

\begin{algorithm}[t]
\caption{Greedy Heuristic for SA-vFPLP}
\KwIn{Network topology $G(N,E)$; RU set $R$; VNC set $V$; wavelength set $W$; $K$-shortest paths between each RU and CN}
\KwOut{$S$ = assignment of VNC, paths, and wavelength to RUs}

$S \leftarrow \emptyset$\;

\SetKwProg{Fn}{Function}{:}{}
\Fn{\textsc{Provision}$(r)$}{
  Compute required backup paths ($\tau$) based on availability requirement\;
  \For{each $v \in V$ in descending priority}{
    \For{each candidate primary path $p$}{
      \If{$p$ is feasible (subject to node capacity, latency, bandwidth constraints)}{
        Find $\tau - 1$ node-disjoint backup paths for $p$\;
        \If{backup paths found}{
          \For{each $\lambda \in W$}{
            \If{$\lambda$ feasible (subject to wavelength constraints)}{
              Assign wavelength $\lambda$ to primary and backup paths\;
              Update resource usage\;
              Add $(r, v, p, \lambda)$ to $S$\;
              \textbf{return}\; 
            }
          }
        }
      }
    }
  }
}

Sort $R$ by slice type priority (URLLC $>$ eMBB $>$ mMTC)\;
\For{each $r \in R$}{
  \textsc{Provision}$(r)$\;
}
\textbf{return} $S$\;
\end{algorithm}

\subsection{Genetic Algorithm}

While the greedy heuristic substantially reduces the computational burden of the ILP, its solution quality can be noticeably sub-optimal on some instances. To improve solution quality, we also propose a metaheuristic -- Genetic Algorithm (GA), an evolutionary optimization method inspired by natural selection \cite{holland1992adaptation}. 

In SA-vFPLP, a chromosome represents a complete assignment for all RUs. Each gene corresponds to one RU and encodes the 5-tuple \texttt{(RU\_id, VNC\_id, primary\_path\_id, backup\_path\_ids, wavelength)}, i.e., the chosen VNC, one primary path, $\tau-1$ node-disjoint backup paths, and a wavelength. Feasibility must satisfy capacity, latency, disjointness, wavelength-availability, and slice-specific constraints. We assume the $K$ shortest paths between each RU and the CN are precomputed and available to the GA.

Initialization draws a population of feasible chromosomes by sampling valid VNCs per RU, selecting one primary path from the precomputed $K$-shortest set, choosing $\tau-1$ node-disjoint backups, and assigning an available wavelength that does not violate conflict rules. This ensures the search starts in the feasible region of SA-vFPLP.

Evolution proceeds by repeatedly applying selection, crossover, and mutation. We use tournament selection to bias reproduction toward fitter chromosomes while preserving diversity. Single-point crossover recombines parental segments because recombination can violate SA-vFPLP semantics (e.g., produce node-disjoint backup paths or oversubscribe node/wavelength capacities). Each offspring is passed through a repair routine that adjusts paths and/or wavelengths to restore feasibility, resampling from the relevant $K$-shortest sets when necessary. Mutation perturbs a randomly chosen RU gene by changing its \texttt{VNC\_id}, path choices, or wavelength; the resulting chromosome is again repaired to maintain feasibility.

Fitness evaluates the same objective as the ILP (eq. 7): the net profit defined as total revenue from admitted slice requests minus node activation, VNF instantiation, and wavelength activation costs, subject to the same feasibility constraints.
The GA iterates until reaching a generation cap or a convergence criterion such as no improvement in best fitness over $g$ consecutive generations.

Compared with greedy heuristics, GA explores a broader portion of the SA-vFPLP search space and is less prone to getting trapped in poor local optima. In practice, it often yields higher-quality solutions. However, the trade-off is increased computational cost and potentially longer convergence times relative to the greedy approach.


\begin{algorithm}[t]
\caption{Genetic Algorithm for SA-vFPLP}
\KwIn{RU set $R$, VNC set $V$, network topology $G(N,E)$, population size $P$, generations $G$}
\KwOut{Best VNC, paths and wavelength assignment solution $S_{best}$}

Initialize population \textit{Pop} with $P$ random feasible solutions\;
\For{$g = 1$ \KwTo $G$}{
  $\mathrm{\textit{NewPop}} \gets \emptyset$\;
  \While{$|\mathrm{\textit{NewPop}}| < P$}{
    Select parents $p_1, p_2$ using tournament selection\;
    $\mathrm{\textit{offspring}} \gets Crossover(p_1, p_2)$\;
    $Mutate(\mathrm{\textit{offspring}})$\;
    Repair \textit{offspring} if infeasible\;
    Add \textit{offspring} to \textit{NewPop}\;
  }
  Evaluate fitness of solutions in \textit{NewPop}\;
  \textit{Pop} $\gets$ \textit{NewPop}\;
  Update $S_{best}$ if improved\;
}
\textbf{return} $S_{best}$\;
\end{algorithm}

\section{Results and Discussion}
\subsection{Experimental Setup}
We implemented our ILP using IBM ILOG CPLEX Optimization Studio v22.1.0. 
The ILP code used the DOcplex library in Python \cite{ibm_docplex_2025}, which is a Python wrapper for the CPLEX studio. 
The heuristic and GA codes are also written in Python. For evaluation, we used a set of realistic network topologies with varying node counts (16, 32, 64, and 128 nodes), derived from the futuristic RAN aligned with the PASSION project \cite{ahmed2024sliavailran} \cite{morais2022placeran} (See Fig.~\ref{fig:topologies}). 
These topologies represent different scales of deployment scenarios for 5G and beyond networks, from small cell deployments to large metropolitan areas \cite{ahmed2024sliavailran}\cite{morais2022placeran}.
\begin{figure*}[htbp]
    \centering
    \begin{subfigure}[b]{0.48\textwidth}
        \centering
        \includegraphics[width=0.9\linewidth]{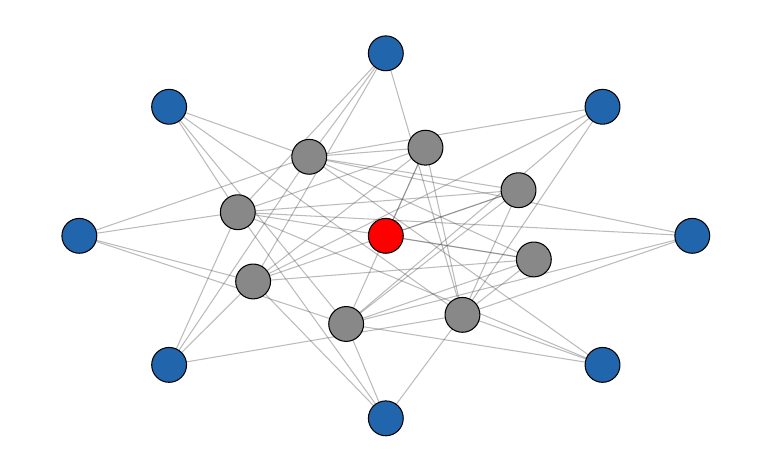}
        \caption{16-Node Topology}
        \label{fig:topo16}
    \end{subfigure}
    \hfill
    \begin{subfigure}[b]{0.48\textwidth}
        \centering
        \includegraphics[width=0.9\linewidth]{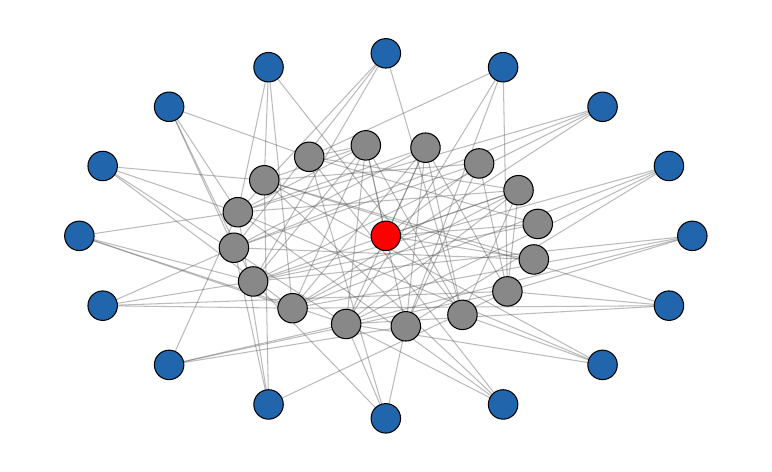}
        \caption{32-Node Topology}
        \label{fig:topo32}
    \end{subfigure}

    \vspace{0.1cm}

    \begin{subfigure}[b]{0.48\textwidth}
        \centering
        \includegraphics[width=\linewidth,height=4cm]{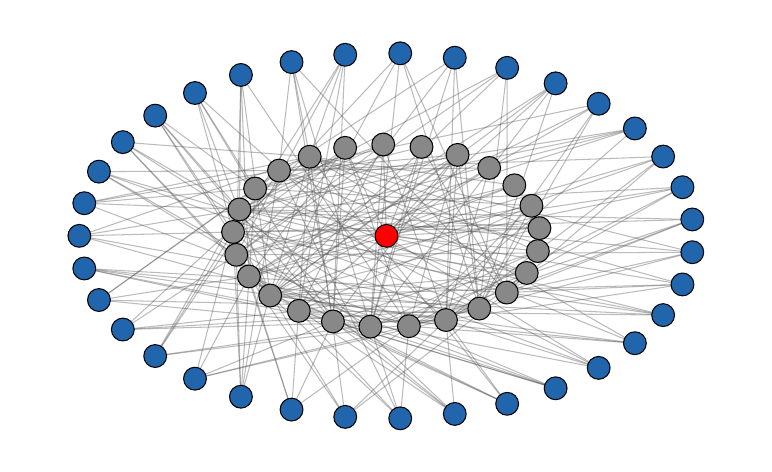}
        \caption{64-Node Topology}
        \label{fig:topo64}
    \end{subfigure}
    \hfill
    \begin{subfigure}[b]{0.48\textwidth}
        \centering
        \includegraphics[width=\linewidth,height=4cm]{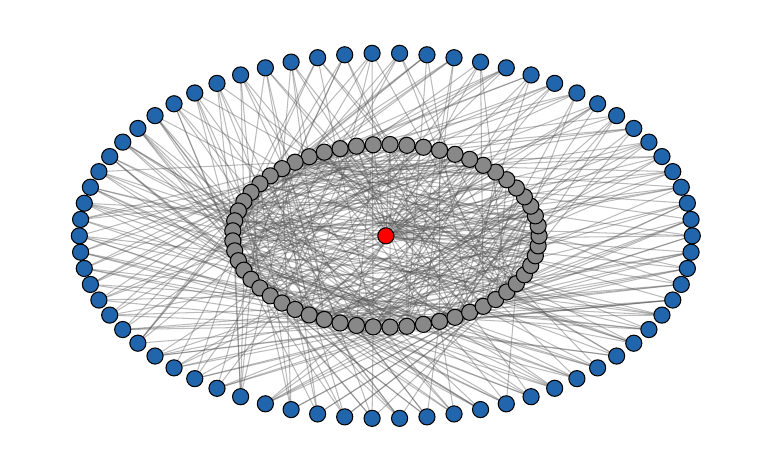}
        \caption{128-Node Topology}
        \label{fig:topo128}
    \end{subfigure}

    \vspace{0.15cm}
    \centering
    \textcolor{blue}{\rule{0.35cm}{0.35cm}}~RU (Radio Unit)\quad
    \textcolor{gray}{\rule{0.35cm}{0.35cm}}~DU/CU (Dist./Central Unit)\quad
    \textcolor{red}{\rule{0.35cm}{0.35cm}}~Core

    \caption{Visualization of different topologies used for experimental 
             evaluation: (a) 16-node, (b) 32-node, (c) 64-node, and (d) 
             128-node. The red node represents the core node, and blue 
             nodes represent RUs.}
    \label{fig:topologies}
\end{figure*}
Each topology comprises three node types: RUs, CUs/DUs, and a single 5G core node. In the 16-, 32-, 64-, and 128-node topologies, the numbers of RUs are 8, 16, 35, and 70, respectively. The remaining nodes in each case (in addition to the single 5G core) serve as candidate CU/DU locations with heterogeneous processing capacities. In Fig. \ref{fig:topologies}, RUs are shown in blue, the 5G core in red, and the CU/DU nodes in grey.

\par In our network model, RU nodes provide 8 CPU units of processing capacity, DU/CU nodes provide 16 CPU units each, and the 5G core node provides 64 CPU units to handle centralized functions. Links from RUs to candidate DUs have a capacity of 50 Gbps; DU–CU links have 100 Gbps and links to the core operate at 800 Gbps (following the PASSION project \cite{passion_project_2020}). Propagation delays on these links range from 0.103 to 0.271 ms, consistent with RAN.

\par Following our prior work \cite{ahmed2024sliavailran}, VNF-instantiation costs are set to 3, 2, and 1 on RU, DU, and CU nodes, respectively, reflecting the decreasing cost of hosting functions as we move toward more centralized nodes in vRAN deployments. The node-activation cost is 1, and the wavelength-activation cost is 1 unit. We assume 40 available wavelength channels in the C-band with 100-GHz (0.8-nm) spacing, as specified by ITU G.694.1 \cite{itu_g6941_2020}. The revenues for accepting URLLC, eMBB, and mMTC slice requests are set to 1000, 500, and 250 units, respectively, reflecting their relative importance and resource demands.
The absolute revenue and cost values are representative planning parameters, not absolute monetary quantities; the optimization behaviour is primarily governed by their relative ratios. The URLLC:eMBB:mMTC revenue ordering reflects the relative SLA strictness of these service classes, while the RU:DU:CU instantiation-cost ordering captures the increasing pooling efficiency toward centralized nodes. Operators can replace these values with deployment-specific business parameters without changing the proposed formulation.
\begin{table}[t]
\centering
\begin{threeparttable}
\caption{One-way Latency and Bandwidth for Different VNCs (Adopted from \cite{morais2022placeran})}
\label{tbl:vnc-parameters}
\setlength{\tabcolsep}{5pt} 
\renewcommand{\arraystretch}{1.15} 
\begin{tabularx}{\columnwidth}{|c|YYY|YYY|}
\hline
\rowcolor[HTML]{E6E6E6}
\textbf{VNC} &
\multicolumn{3}{c|}{\textbf{One-way Latency (ms)\tnote{1}}} &
\multicolumn{3}{c|}{\textbf{Bandwidth (Gbps)\tnote{2}}} \\
\cline{2-7}
\rowcolor[HTML]{E6E6E6}
 & \textbf{BH} & \textbf{MH} & \textbf{FH} & \textbf{BH} & \textbf{MH} & \textbf{FH} \\
\hline
9 & 1.5$\sim$10 & 1.5$\sim$10 & 0.25 & 9.9 & 13.2 & 42.6 \\
8 & 1.5$\sim$10 & 1.5$\sim$10 & 0.25 & 9.9 & 13.2 & 42.6 \\
7 & 1.5$\sim$10 & 1.5$\sim$10 & 0.25 & 9.9 & 13.2 & 13.6 \\
6 & 1.5$\sim$10 & 1.5$\sim$10 & 0.25 & 9.9 & 13.2 & 13.6 \\
5 & 1.5$\sim$10 & 1.5$\sim$10 & --   & 9.9 & 13.2 & --   \\
4 & 1.5$\sim$10 & 1.5$\sim$10 & --   & 9.9 & 13.2 & --   \\
3 & 1.5$\sim$10 & --          & 0.25 & 9.9 & --   & 13.6 \\
2 & 1.5$\sim$10 & --          & 0.25 & 9.9 & --   & 42.6 \\
1 & 1.5$\sim$10 & --          & --   & 9.9 & --   & --   \\
\hline
\end{tabularx}

\begin{tablenotes}[flushleft]
\footnotesize
\item[1] Maximum tolerable one-way latency.
\item[2] Taken in accordance with the PASSION Project \cite{passion_project_2020}.
\end{tablenotes}
\end{threeparttable}
\end{table}
\par The availability of each processing node was set to 0.999, which is consistent with typical reliability values for network equipment. For different slice types, we maintained the following availability requirements: URLLC slices with availability requirements randomly chosen from [0.9999, 0.99999], mMTC from [0.95, 0.999], and eMBB from [0.99, 0.999] \cite{ahmed2024sliavailran}. 
Parameters associated with different VNCs, such as CPU usage at nodes, latency, and bandwidth requirements at the backhaul, fronthaul, and midhaul links for each VNC, were adopted from our previous work \cite{ahmed2024sliavailran}. 
Table \ref{tbl:vnc-parameters} provides the details of the same. The bandwidth allocation costs for fronthaul, midhaul, and backhaul were set at 1 unit.
For each topology, we generated multiple batches of slice requests with different slice-type mixes and evaluated our framework. 

\subsection{Results and Evaluation}

\subsubsection{\textbf{Slice acceptance}}
Fig. \ref{fig:rar_all} shows the request acceptance rate, i.e., the percentage of total slice requests successfully admitted. We observe that the shared scheme achieves a higher or equivalent acceptance compared to its unshared counterpart, for all settings. This is because sharing backup VNF instances whenever possible frees compute capacity that can accommodate additional requests. 
This gap between shared and unshared versions is more pronounced in the 32-node topology, reaching 4.1\% for ILP, 6.7\% for heuristic, and 8.3\% for GA. At 16-nodes, both shared and unshared versions converge to 85.2\% in the case of ILP. 
\begin{figure*}[t]
  \centering
  \begin{subfigure}[b]{0.325\textwidth}
    \includegraphics[width=\textwidth]{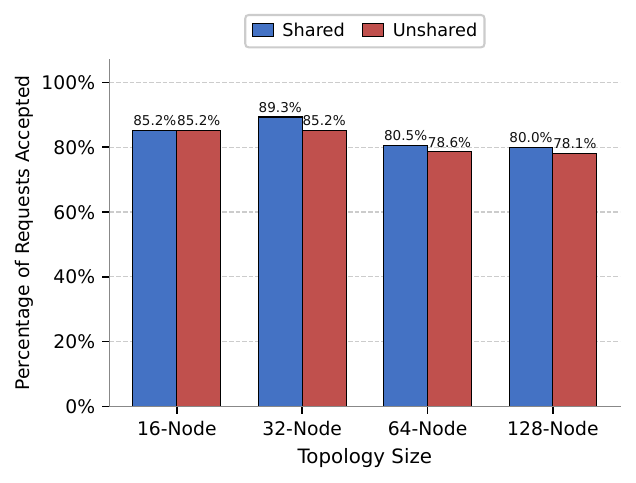}
    \caption{ILP Algorithm}
    \label{fig:rar_ilp}
  \end{subfigure}
  \hfill
  \begin{subfigure}[b]{0.325\textwidth}
    \includegraphics[width=\textwidth]{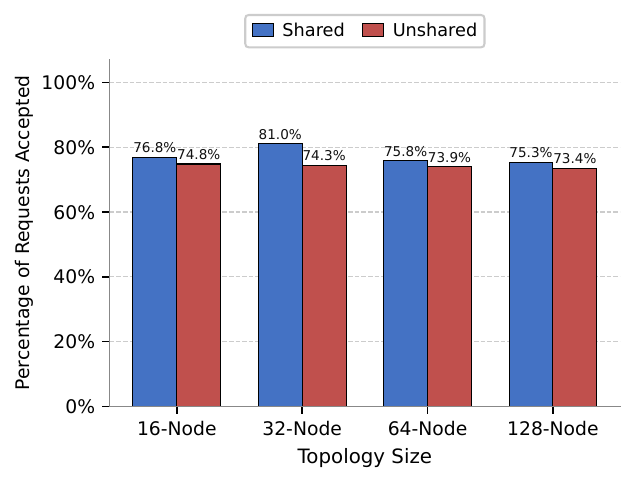}
    \caption{Heuristic Algorithm}
    \label{fig:rar_heuristic}
  \end{subfigure}
  \hfill
  \begin{subfigure}[b]{0.325\textwidth}
    \includegraphics[width=\textwidth]{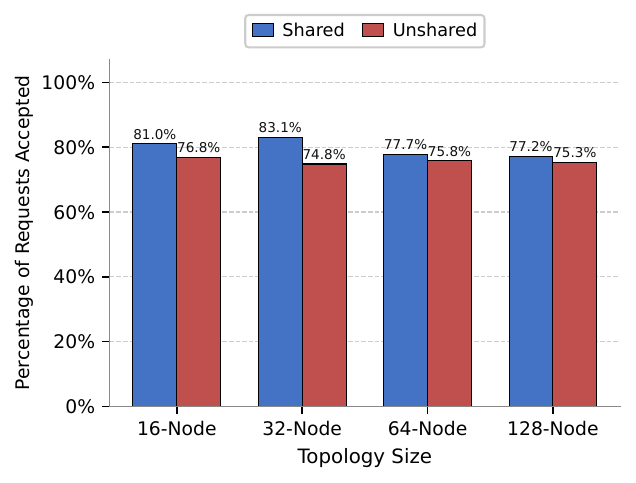}
    \caption{Genetic Algorithm}
    \label{fig:rar_ga}
  \end{subfigure}
  \caption{Request Acceptance Rate (\%) across shared and unshared variants 
           of ILP, Heuristic, and GA algorithms for different topology sizes.}
  \label{fig:rar_all}
\end{figure*}
Also, we see that the acceptance rate slightly increases at 32 nodes and then decreases and stabilizes at 64 and 128 nodes, for all algorithms. 
We have fewer DU/CU nodes at 16-nodes, leading to a limited path diversity. This makes it harder for higher availability slice requests, like URLLC, to satisfy node-disjoint backup paths requirements. 
The ratio of DU/CU nodes to RUs is more favorable for 32-nodes, providing a better path diversity while keeping the total volume of competing requests manageable, leading to higher acceptance rates across all settings. 
For topologies beyond 32 nodes, the number of slice requests grows proportionally with the number of RUs (35 RUs at 64 nodes, 70 at 128 nodes). Thus, despite the availability of seemingly more physical resources, the cumulative demand for primary and backup lightpaths intensifies resource contention. This causes the acceptance to slightly decline before stabilizing. 
Across all topologies, ILP consistently shows higher acceptance (e.g., 80\% for shared at 128 nodes), followed by GA (77.2\%) and Heuristic (75.3\%). 

\subsubsection{\textbf{VNC assignment}}

As detailed in Section \ref{section:vnc}, higher VNC IDs indicate higher priority configurations based on practical deployment trade-offs (centralization vs. transport requirements vs. feasibility). 
Specifically, NG-RAN-III (VNCs 6–9) is highest-priority; C-RAN (VNCs 4–5), though more centralized, ranks below NG-RAN-III due to stringent fronthaul demands; NG-RAN-II (VNCs 2–3) and D-RAN (VNC 1) have the lowest priority as they retain most functions at the edge and forgo pooling gains.

In our evaluation, the VNC distribution (Fig. \ref{fig:vnc-distribution-grid}) shows clear allocation patterns that vary with topology size and algorithm choice.
For the 16-node topology, ILP (unshared) concentrates placements in higher-priority VNCs (VNC~8: 37.5\%, VNC~7: 25\%, VNC~6: 37.5\%). ILP (shared) shifts further upward (VNC~9: 25\%, VNC~8: 62.5\%, VNC~6: 12.5\%). The heuristic shows greater dispersion: the \textit{unshared} form assigns 62.5\% to lower-priority VNCs (VNC~5 and below) and 25\% to VNC~9 (with the remaining 12.5\% in mid-tier VNCs), whereas the shared form tilts toward higher-priority options (VNC~9: 37.5\%, VNC~8: 12.5\%, VNC~7: 12.5\%, and 37.5\% to VNCs~5, 3, and 2). The GA follows a similar pattern: unshared (VNC~9: 25\%, VNC~7: 37.5\%, VNC~4: 37.5\%); shared (VNC~9: 37.5\%, VNC~7: 37.5\%, VNC~6: 25\%).

In the 32-node topology, ILP (unshared) concentrates 62.5\% in VNC~6, whereas ILP (shared) redistributes toward higher-priority options -- VNC~7 (31.3\%), VNC~8 (37.5\%), and VNC~9 (12.5\%). As the network grows, allocations shift upward: at 64 nodes both ILP variants reduce VNC~6 (25.7\% unshared; 20\% shared) and increase VNC~9 (34.3\% vs.\ 31.4\%). 
At 128 nodes, both continue to favor higher-priority VNCs (greater mass on VNCs~4--9). 
GA and heuristic show a similar trend in larger topologies; for example, at 128 nodes, GA (unshared) places 38.6\% in VNC~6 and 32.9\% in VNC~7, while GA (shared) assigns 50\% to VNC~7 and 32.9\% to VNC~9 with minimal use of lower VNCs. 
In smaller topologies (16/32 nodes), the heuristic and GA allocate substantial shares to lower VNCs, while ILP places nearly all requests above VNC~5, even if the absolute share of VNC~9 remains modest. Overall, ILP (especially with shared backup) drives placements toward higher-priority configurations. The heuristic and GA remain more dispersed across VNCs.

\begin{figure*}[htbp]
    \centering

    \begin{subfigure}[t]{0.48\textwidth}
        \centering
        \includegraphics[width=\textwidth]{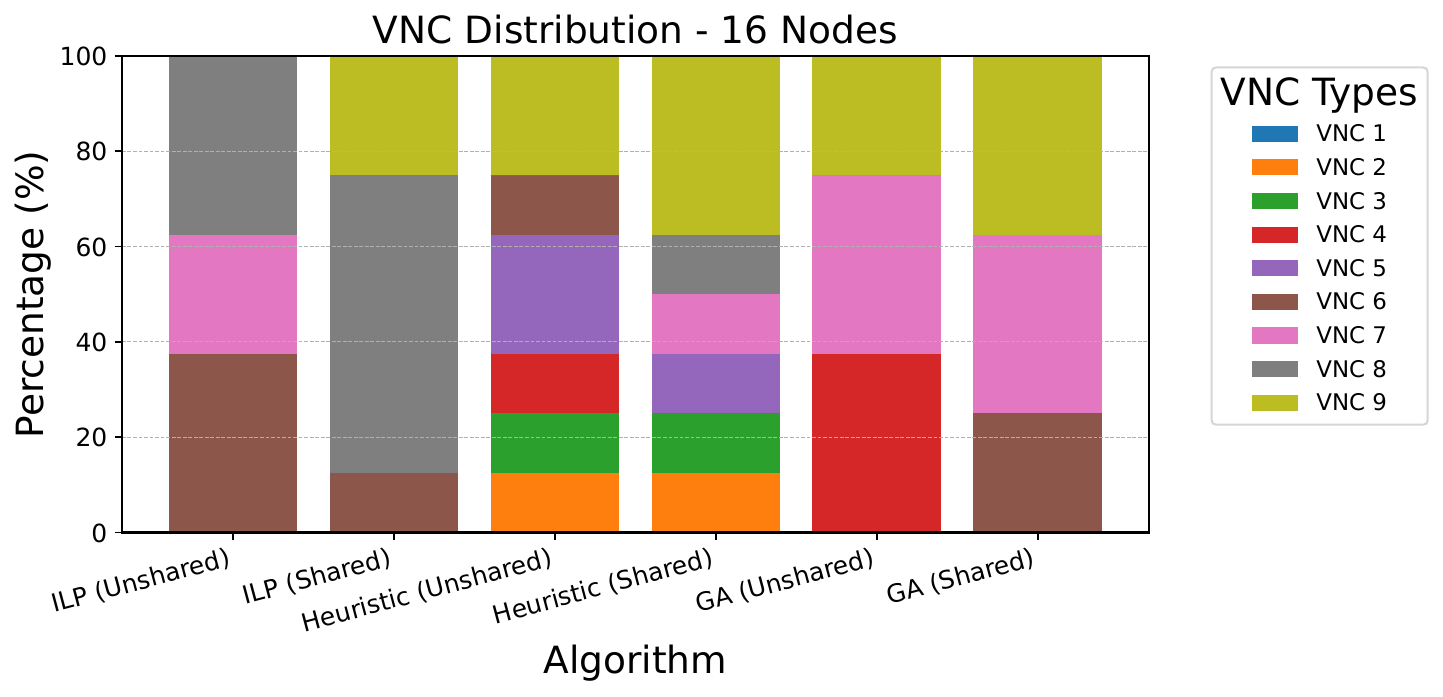}
        \caption{16-node topology}
        \label{fig:vnc-16}
    \end{subfigure}
    \hfill
    \begin{subfigure}[t]{0.48\textwidth}
        \centering
        \includegraphics[width=\textwidth]{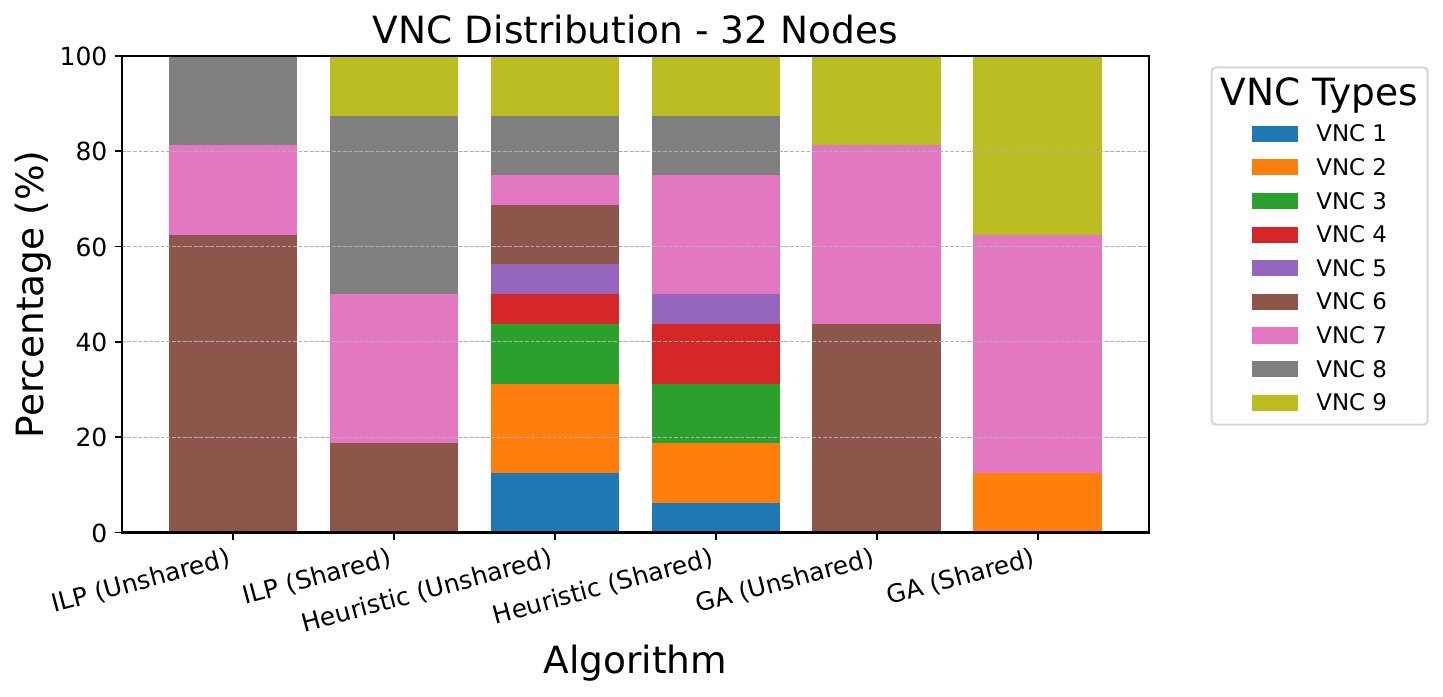}
        \caption{32-node topology}
        \label{fig:vnc-32}
    \end{subfigure}

    \vspace{0.3cm}

    \begin{subfigure}[t]{0.48\textwidth}
        \centering
        \includegraphics[width=\textwidth]{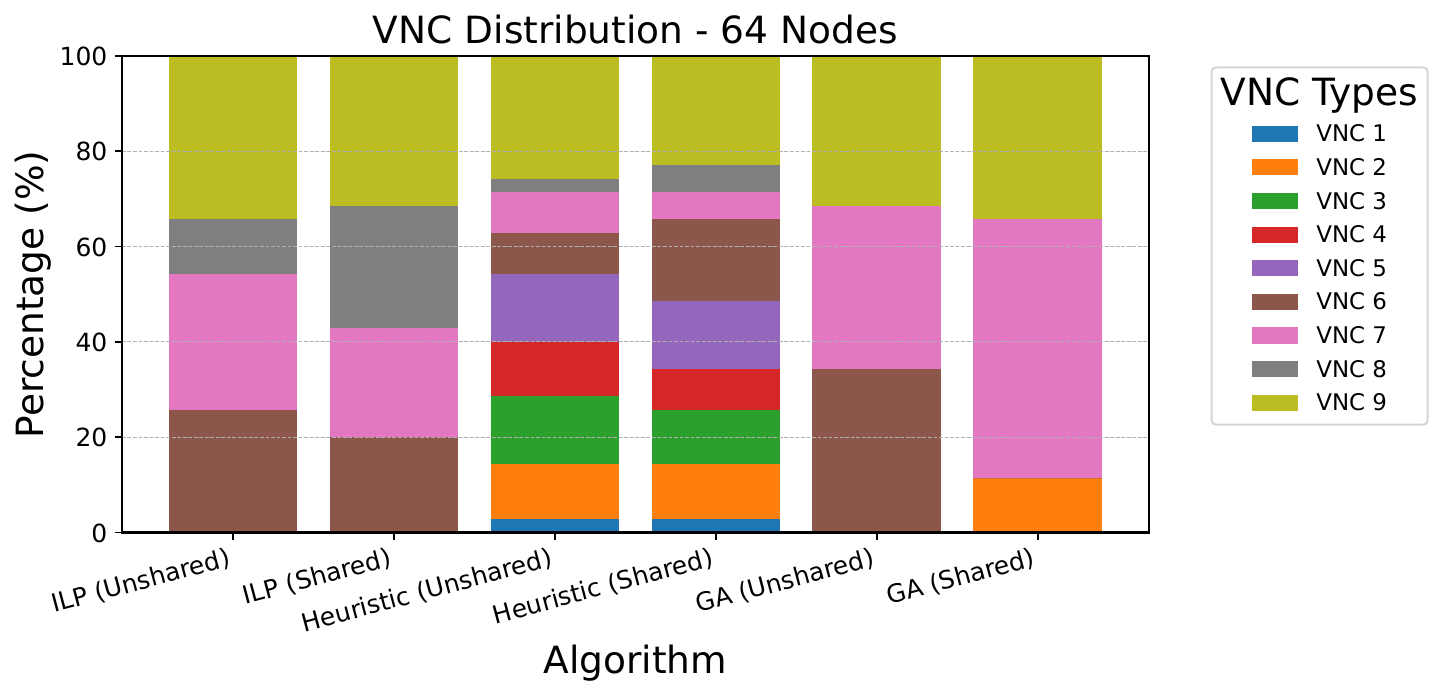}
        \caption{64-node topology}
        \label{fig:vnc-64}
    \end{subfigure}
    \hfill
    \begin{subfigure}[t]{0.48\textwidth}
        \centering
        \includegraphics[width=\textwidth]{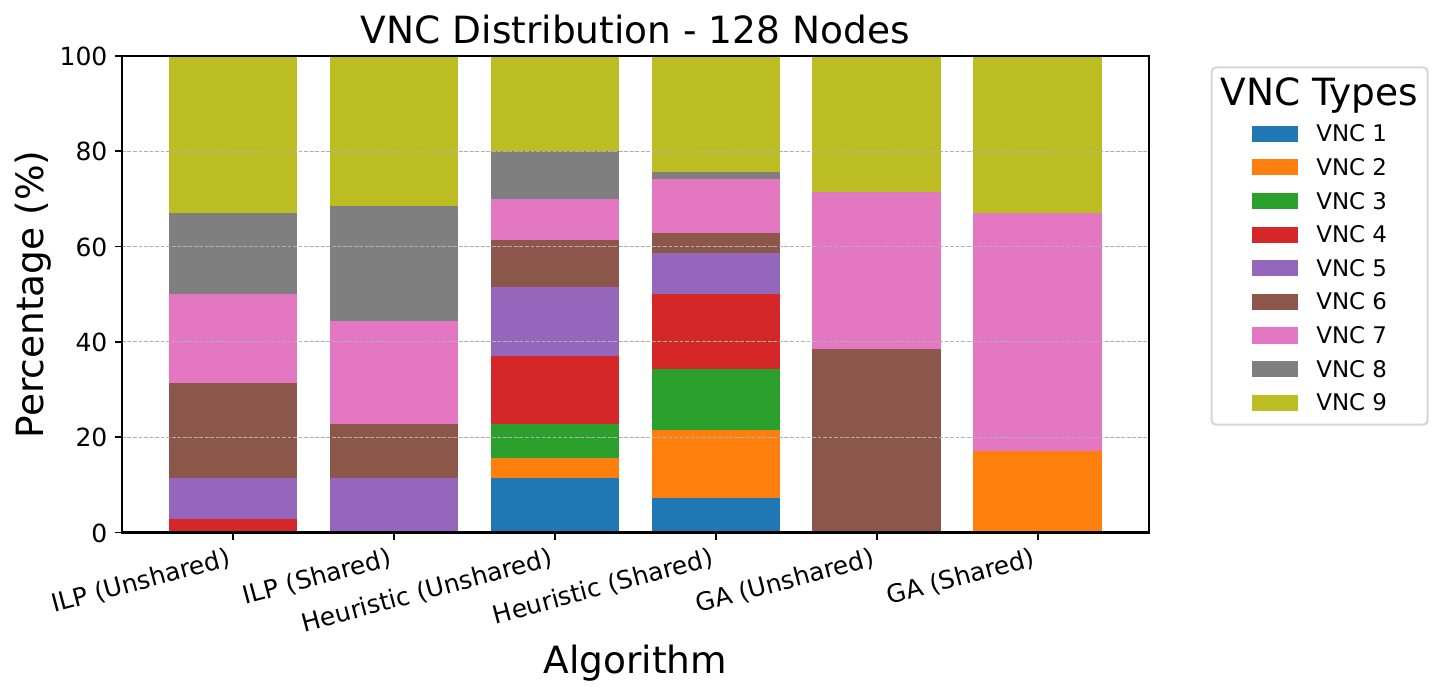}
        \caption{128-node topology}
        \label{fig:vnc-128}
    \end{subfigure}

    \caption{VNC distribution (\%) across shared and unshared variants of ILP, Heuristic, and GA algorithms for different topologies. Each bar is stacked with nine segments representing the percentage of slice requests mapped to the respective VNCs.}
    \label{fig:vnc-distribution-grid}
\end{figure*}

\subsubsection{\textbf{Profit}} We now report the profit because it is the \emph{primary objective} of our formulation: the ILP maximizes the service–provider margin, subject to capacity and delay constraints. Profit thus captures both \emph{admission/revenue} (which slices are accepted) and \emph{placement efficiency/cost}. Consistent with this objective, the total profit increases with topology size across all algorithms and variants, as larger networks admit more slice requests and exploit the pooling to amortize transport and compute costs. 

Figure \ref{fig:profit-slicetypes} illustrates these trends for total profit values for all the algorithms across all topology sizes, for both unshared and shared variants. We analyze the profits across four distinct slice distributions with (i) all slice requests as URLLC, (ii) all as eMBB,  (iii) all as mMTC, and (iv) a balanced distribution of one-third of each slice type, to examine how the distribution of slice types affects the overall profitability of our solutions across these configurations.  Among all the topologies and algorithms, the highest profit was generated by the URLLC-only configuration, followed by the balanced mix, all eMBB, and mMTC. 
The observations align with our cost model, where URLLC slice requests have been assigned the highest revenue amongst the other two request types, thus increasing the overall profit considerably for all URLLC type requests. For example, in the 128-node case, the ILP-shared version for this workload generates 98.3\% more profit than the all-eMBB version and almost four times more than the all-mMTC version.
Additionally, shared backup variants outperformed unshared ones in all configurations, with ILP shared and GA shared achieving higher total profits consistently. 
For example, in the 128-node all-URLLC scenario, the shared scheme yields 16.7\% higher profit for the ILP and 18.0\% for the GA compared to their unshared counterparts, demonstrating the cost-effectiveness of backup sharing.
By adopting a shared backup scheme, the system can admit more high-revenue slices under constrained bandwidth and compute resources. This pattern also holds in the equal slice distribution case, where ILP-shared and GA-shared consistently achieve higher profits.
These results establish the ability of our framework to adapt to heterogeneous request distributions and still allocate requests in alignment with the profitability objectives of the mobile network operators (MNOs). 



\begin{figure*}[htbp]
    \centering
    \begin{subfigure}[t]{0.48\textwidth}
        \includegraphics[width=\textwidth]{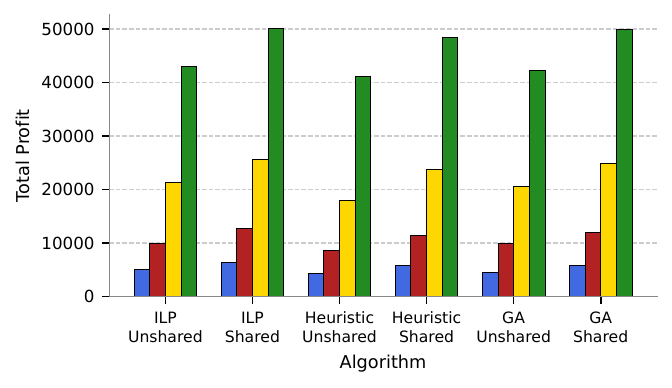}
        \caption{All uRLLC Requests}
    \end{subfigure}
    \hfill
    \begin{subfigure}[t]{0.48\textwidth}
        \includegraphics[width=\textwidth]{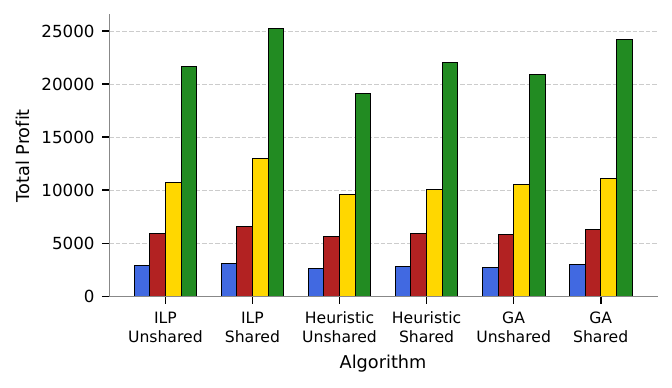}
        \caption{All eMBB Requests}
    \end{subfigure}

    \vspace{0.1cm}

    \begin{subfigure}[t]{0.48\textwidth}
        \includegraphics[width=\textwidth]{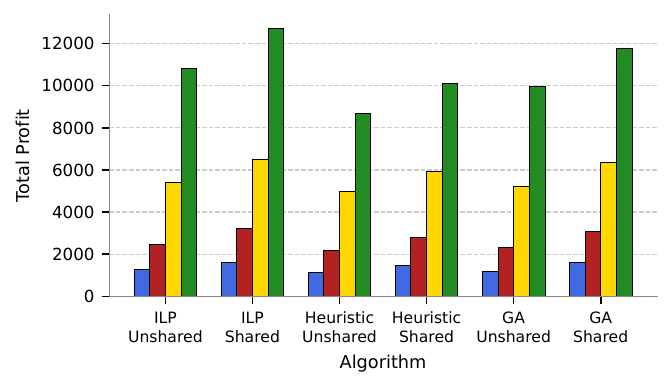}
        \caption{All mMTC Requests}
    \end{subfigure}
    \hfill
    \begin{subfigure}[t]{0.48\textwidth}
        \includegraphics[width=\textwidth]{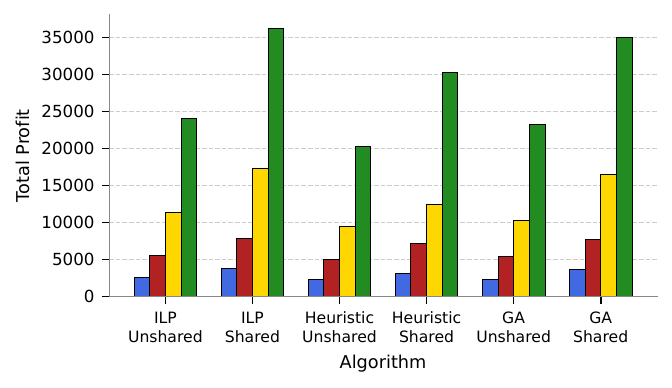}
        \caption{Equal Distribution of Slice Types}
    \end{subfigure}

    \vspace{0.15cm}
    \centering
  \textcolor{myblue}{\rule{0.35cm}{0.35cm}}~16-Nodes\quad
\textcolor{myred}{\rule{0.35cm}{0.35cm}}~32-Nodes\quad
\textcolor{mygold}{\rule{0.35cm}{0.35cm}}~64-Nodes\quad
\textcolor{mygreen}{\rule{0.35cm}{0.35cm}}~128-Nodes
    \caption{Total profit across ILP, Heuristic, and GA algorithms for 
             various request slice-type distributions over increasing 
             topology sizes.}
    \label{fig:profit-slicetypes}
\end{figure*}

\begin{figure*}[htbp]
    \centering
    \begin{subfigure}[t]{0.32\textwidth}
        \centering
        \includegraphics[width=\linewidth]{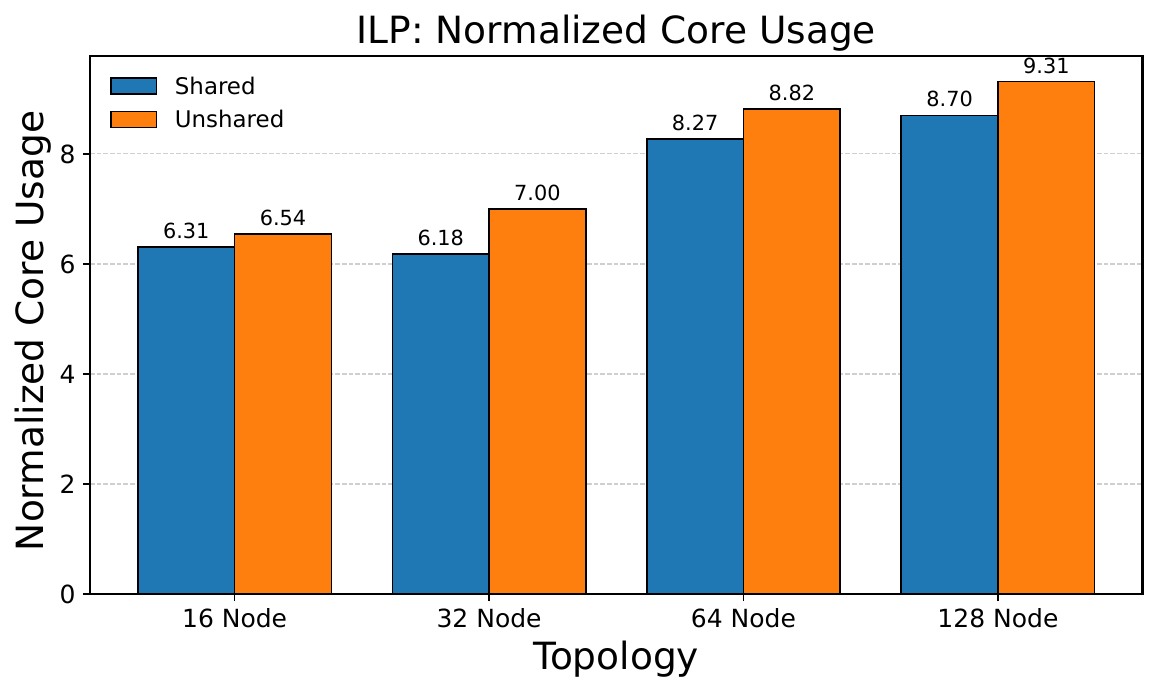}
        \caption{ILP Algorithm}
    \end{subfigure}
    \hfill
    \begin{subfigure}[t]{0.32\textwidth}
        \centering
        \includegraphics[width=\linewidth]{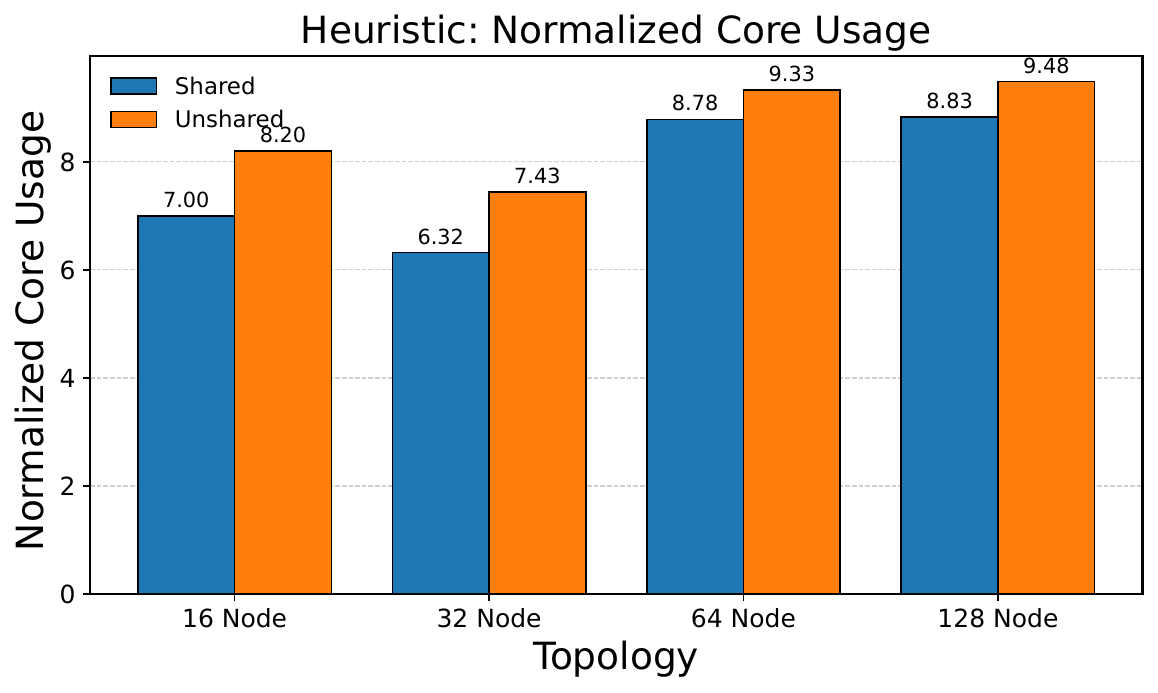}
        \caption{Heuristic Algorithm}
    \end{subfigure}
    \hfill
    \begin{subfigure}[t]{0.32\textwidth}
        \centering
        \includegraphics[width=\linewidth]{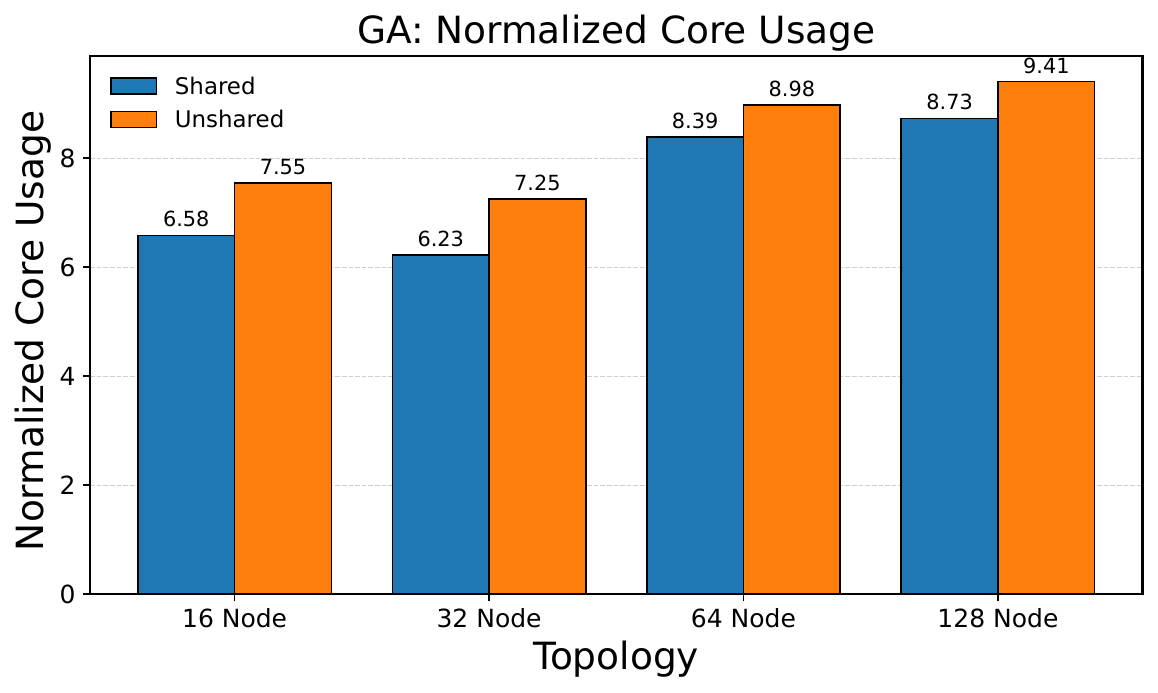}
        \caption{Genetic Algorithm}
    \end{subfigure}
    \caption{Normalized core usage across shared and unshared variants of ILP, Heuristic, and GA algorithms for increasing topology sizes.}
    \label{fig:cpu-utilization}
\end{figure*}

\subsubsection{\textbf{Normalized CPU (core) and link usage}}

Table \ref{tbl:core utilization} provides the core usage, indicating the required CPU core count to execute VNFs for each protocol layer, values derived from the OAI implementation \cite{Kobayashi2017OAI} \cite{chun2016performance}). Core usage is an important metric for evaluating the resource efficiency of all our approaches and their shared and unshared versions. We normalize the core usage based on the request acceptance for a fair comparison, as the number of accepted requests varies considerably among all the algorithms. Here, a lower value of the Normalized Core Usage (NCU) indicates more efficient use of computational resources.

As shown in Fig.~\ref{fig:cpu-utilization}, shared backup schemes consistently achieve a lower NCU than their unshared counterparts across all topologies and algorithms (for instance, for ILP, the NCU ranges from 6.31 for shared vs. 6.54 for unshared for 16 nodes, to 8.70 shared vs. 9.31 unshared for 128 nodes, similarly for others). This is because the shared schemes allow the sharing of common backup VNF instances among multiple slice requests with compatible VNCs (that share a common backup path). It thus avoids redundant VNF instantiations, thereby reducing the total cores consumed for a given number of accepted requests. Secondly, across increasing topology sizes, ILP exhibits the lowest NCU, followed by GA and then the heuristic. For instance, for the unshared version in 32-nodes, ILP has an NCU of 7, GA has 7.25, and the heuristic has 7.43, respectively, similarly across all topologies and the shared version. 

\begin{table}[h!]
\centering
\begin{tabular}{lc}
\hline
\rowcolor[HTML]{E6E6E6}
\textbf{RAN Protocol} & \textbf{CPU Usage (cores)} \\
\hline
RRC and PDCP        & 0.49 each \\
High RLC and Low RLC   & 0.0245 each\\
High MAC and Low MAC  & 0.343 each\\
High PHY   & 0.833 \\
Low PHY    & 2.352 \\
\hline
\textbf{Total} & \textbf{4.9} \\
\hline
\end{tabular}
\caption{RAN protocol CPU usage (in cores) (Adapted from PlaceRAN \cite{morais2022placeran}, based on data from \cite{Kobayashi2017OAI})}
\label{tbl:core utilization}
\end{table}

Fig. \ref{fig:link-utilization} shows the results for link usage (in percent), calculated as the ratio of total bandwidth consumed to their overall available link capacity, averaged across all active/used links in the topology. The following trends are observed here. Firstly, shared backup variants consistently incur higher link usage than their unshared counterparts, across all algorithms and topologies. For example, ILP records 63.83\%  link usage in the 32-node shared case versus 51.30\% in the unshared variant. While in the 64-node topology, the difference rises to 24.29 percentage points (74.18\% shared vs. 49.89\% unshared). This can be attributed to the fact that the shared version tends to select more centralized VNCs. Additionally, for more centralized VNCs, the bandwidth requirements are considerably higher, particularly for the fronthaul, as shown in Table \ref{tbl:vnc-parameters}. 

Secondly, we observe that ILP consistently shows the highest link usage (e.g., up to 78.01\% in the 16-node shared case) because it tries to assign more centralized VNCs to the requests, which allocate a larger portion of the RAN functions to CU/DU nodes, thereby naturally increasing total link usage. GA occupies the middle ground, with its metaheuristic trying to get near-optimal configurations, resulting in intermediate link usage values (i.e., 55.81\% in the 64-node shared case). The heuristic method records the lowest usage (as low as 40.89\% in the 64-node unshared case), since its locally optimal, greedy choices do not lead to aggressive centralization and keep more VNFs closer to RUs, reducing bandwidth demand but also missing potential pooling gains.

Lastly, with the increase in the topology size, link usage follows a non-monotonic pattern.  It is high at 16 nodes due to the limited number of links, which constricts routing path alternatives and leads to higher congestion on available optical paths (i.e., for eg, ILP-shared: 78.0\%, GA-shared: 66.4\%, Heuristic-shared: 62.5\% for the 16-node case). There is a reduction at 32 nodes when more links and alternative paths become available, momentarily reducing congestion. It rises again at 64 nodes because the increase in the number of accepted slice requests, combined with the selection of more centralized VNCs, collectively increases total bandwidth demand. At 128 nodes, the link usage begins to stabilize as additional optical capacity compensates for the increased number of active flows (i.e., ILP-shared: 71.9\%, GA-shared: 64.1\%, Heuristic-shared: 58.9\%). This fluctuating trend can be attributed to the interaction between topology size (due to routing path availability), number of requests accepted, and the centralization level of the selected VNCs.

\begin{figure*}[htbp]
    \centering
    \begin{subfigure}[t]{0.32\textwidth}
        \centering
        \includegraphics[width=\linewidth]{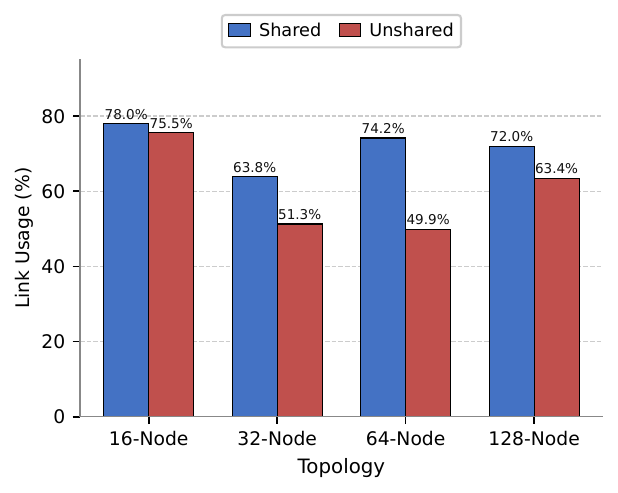}
        \caption{ILP Algorithm}
    \end{subfigure}
    \hfill
    \begin{subfigure}[t]{0.32\textwidth}
        \centering
        \includegraphics[width=\linewidth]{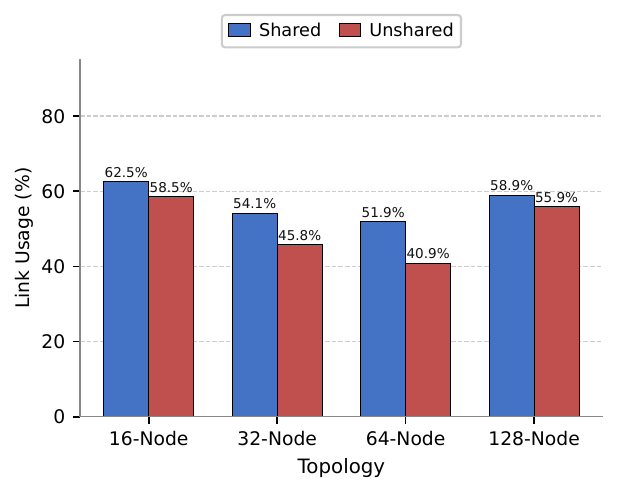}
        \caption{Heuristic Algorithm}
    \end{subfigure}
    \hfill
    \begin{subfigure}[t]{0.32\textwidth}
        \centering
        \includegraphics[width=\linewidth]{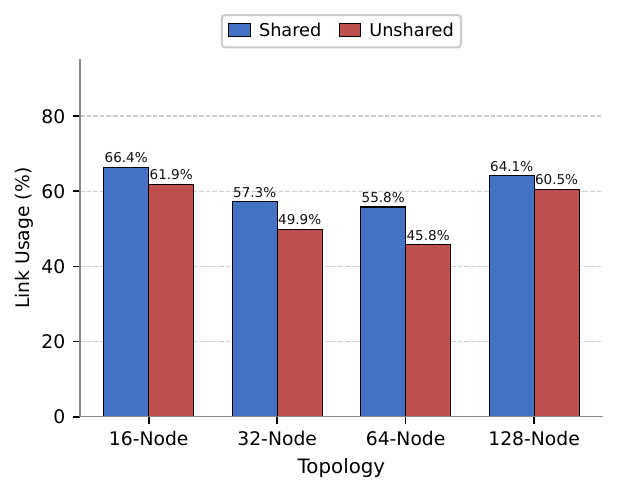}
        \caption{Genetic Algorithm}
    \end{subfigure}

    \caption{Link Usage (\%) across shared and unshared variants of ILP, Heuristic, and GA algorithms over different topologies.}
    \label{fig:link-utilization}
\end{figure*}

\subsubsection{\textbf{Genetic algorithm micro-benchmarking}}
We assess GA sensitivity to \emph{population size} using total profit (Fig.~\ref{fig:population size vs profit}). Across all topologies, profit exhibits clear \emph{diminishing returns} as population increases. For example, in the 128-node \emph{shared} variant, profit rises from 13,502.6 at population~10 to $\approx$20,569 at population~100 (+52.3\%), but improves only to 21,309.5 at population~250 (+3.5\% over 100). Similar curves appear in smaller topologies: the bulk of gains occur for 10--100 individuals, with marginal improvements beyond $\sim$150, indicating practical convergence and justifying population sizes near 150 as a cost-effective choice.

\begin{figure}[htbp]
\centering
	\includegraphics[width=8.5cm, height=5.5cm]{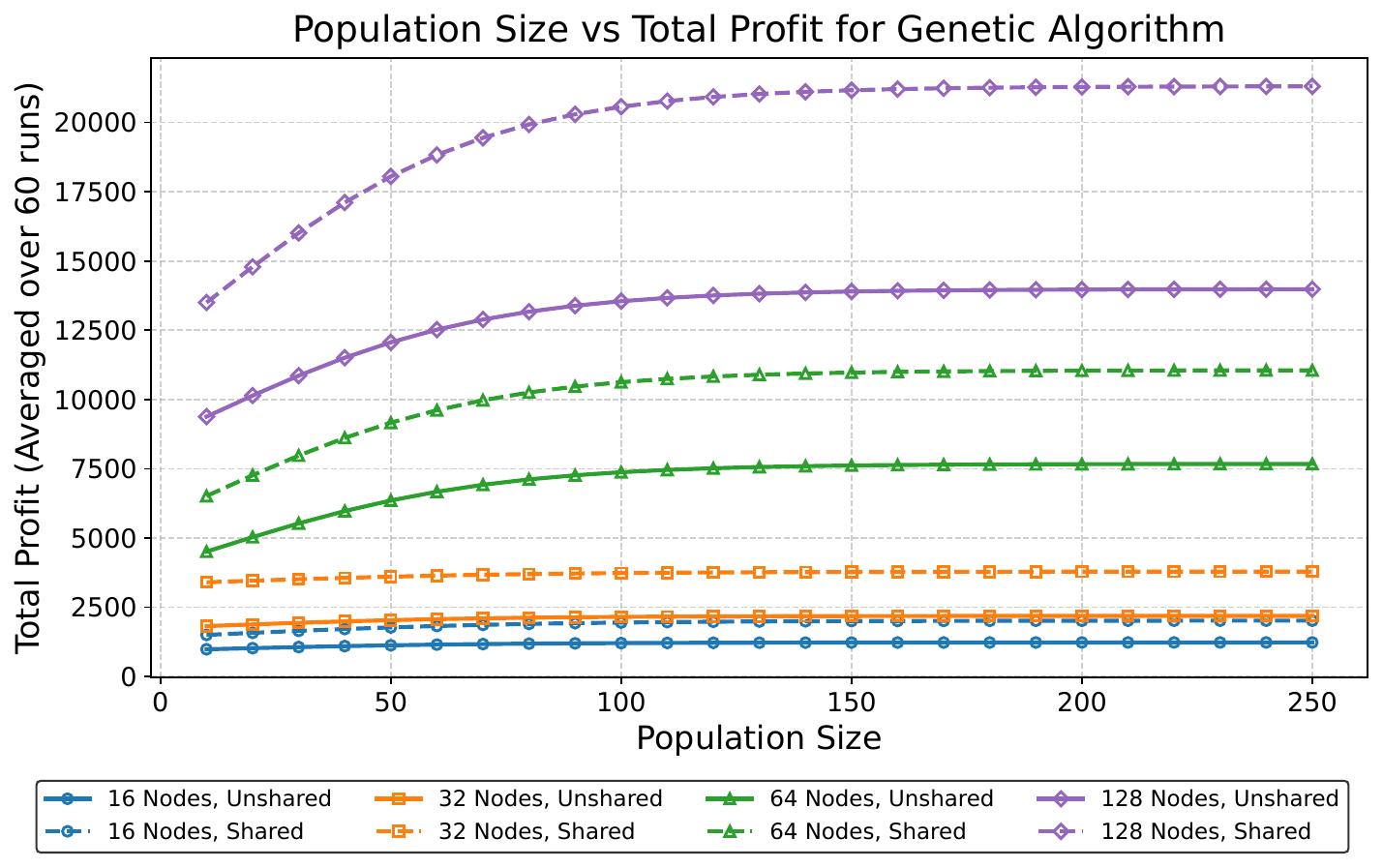}
	\caption{Population size vs total profit for different topologies in shared and unshared versions of the Genetic Algorithm}
	\label{fig:population size vs profit}
\end{figure}


Fig.~\ref{fig:ga-profit-vs-generations} shows the effect of GA \emph{generations} on total profit for both shared and unshared variants across all topologies. The curves consistently converge: in the 128-node (shared) case, profit increases from $\sim$14,000 at 10 generations to $\sim$20,360 at 40 generations (+45.4\%), but only reaches 21,304 at 100 generations. Across topologies, the slope flattens after 60--70 generations, indicating saturation of the search. Based on these results, we set GA to 40 generations with a population size of 100 for the final experiments, balancing computational cost and solution quality. 
Higher parameter settings provide only marginal performance improvements while incurring disproportionately higher runtime, as the profit curves approach saturation. For example, in the 128-node (shared) scenario, increasing the population size from 100 to 250 and the number of generations from 40 to 100 results in only 3.5\% and 4.6\% additional profit, respectively. This indicates that the selected GA configuration achieves more than 95\% of the best observed profit while requiring substantially fewer evaluations.

\begin{figure}[htbp]
    \centering
    \includegraphics[width=9cm, height=5.5cm]{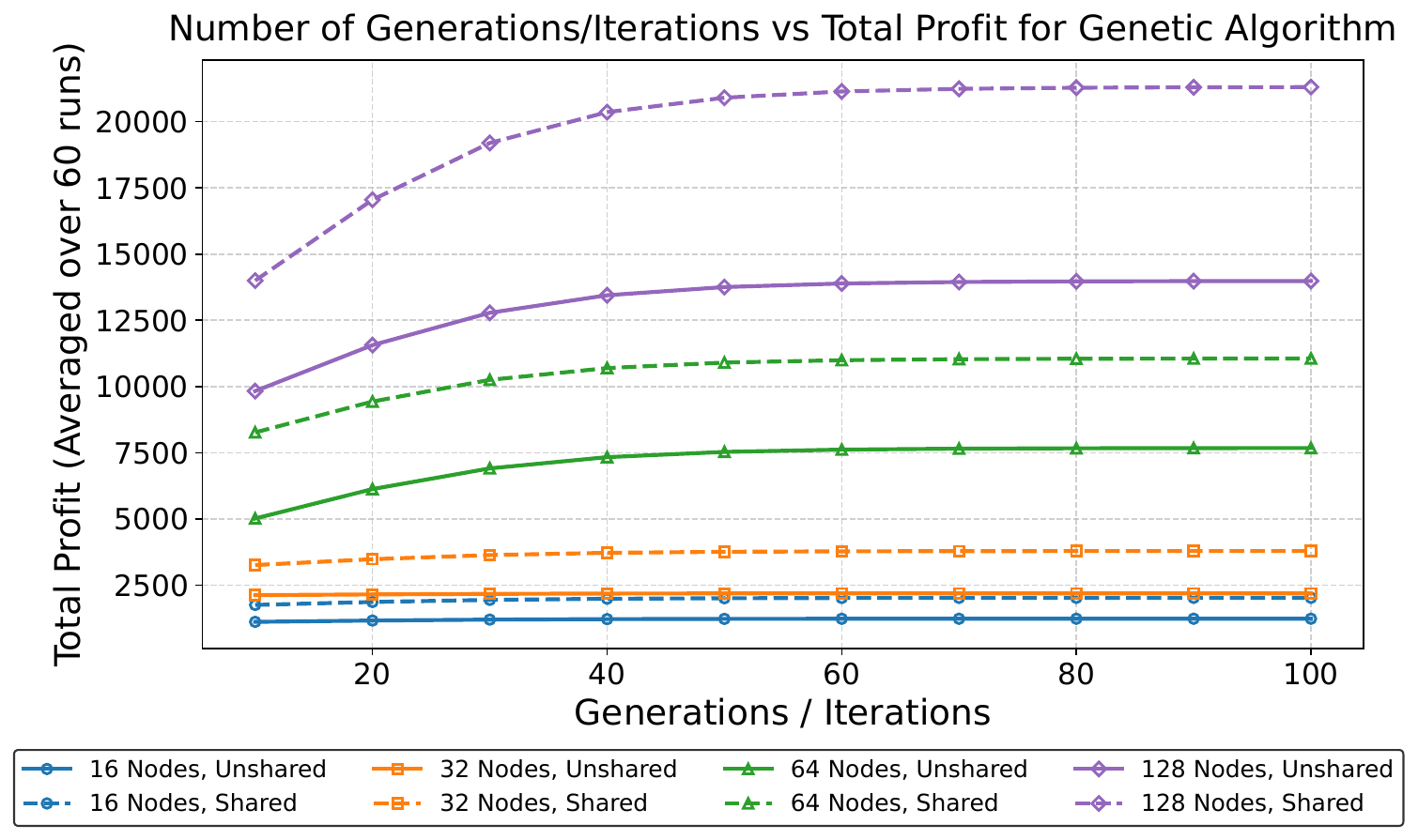}
    \caption{Evolution of total profit over generations for different topology sizes and GA variants (Shared vs. Unshared). Profit values are averaged over 60 runs.}
    \label{fig:ga-profit-vs-generations}
\end{figure}
\subsubsection{\textbf{Wavelength channels}}

We assess optical efficiency via the number of unique wavelengths ($\lambda$) allocated under DWDM constraints (Fig.~\ref{fig:wavelength-channels}). As topology size increases (16$\rightarrow$128 nodes), wavelength usage grows steadily, reflecting greater routing diversity and longer paths. The unshared heuristic consistently consumes the most spectrum, peaking at 36 channels on the 128-node topology. This spectral fragmentation is because locally greedy choices miss opportunities for path/wavelength reuse.
By contrast, ILP–shared and GA–shared variants use substantially fewer wavelengths. 
Two factors drive this: (i) in the shared-backup scheme, multiple RUs with the same VNC can reuse backup paths, which reduces wavelength channels, and (ii) the objective includes a wavelength-activation penalty, $C_{\lambda}$, which encourages reuse of already activated wavelengths and discourages new activations. 
Overall, ILP and GA with shared backup strike a better balance between spectrum efficiency and computational cost, whereas the heuristic trades speed for higher wavelength consumption.


\begin{figure}[htbp]
    \centering
    \includegraphics[width=0.5\textwidth]{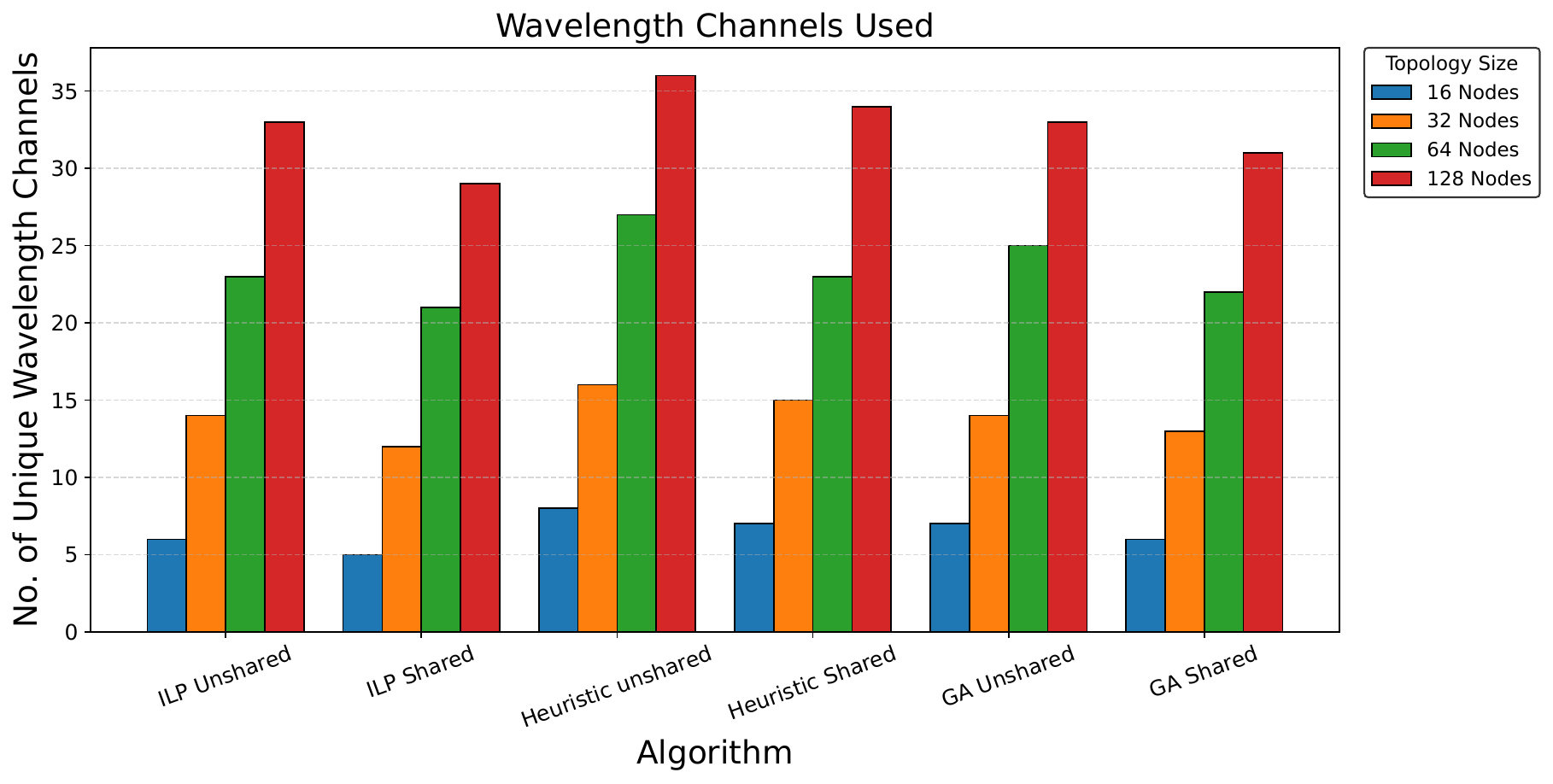}
    \caption{Comparison of number of unique wavelength channels used by different algorithms across topology sizes.}
    \label{fig:wavelength-channels}
\end{figure}

\subsubsection{\textbf{Execution time}} The results for the execution times of different algorithms highlight drastic differences between the exact (ILP) and the approximate (Heuristic and GA) approaches across all topologies, as shown in Fig. \ref{fig:execution-time}. 
The ILP implementations show exponential growth in execution times as the network size increases. The execution times for ILP approaches range from 9.5 seconds for the 16-node topology to over 87 minutes (5208.2 seconds) for the 128-node topology. 
\begin{figure}[htbp]
    \centering
    \includegraphics[width=0.5\textwidth]{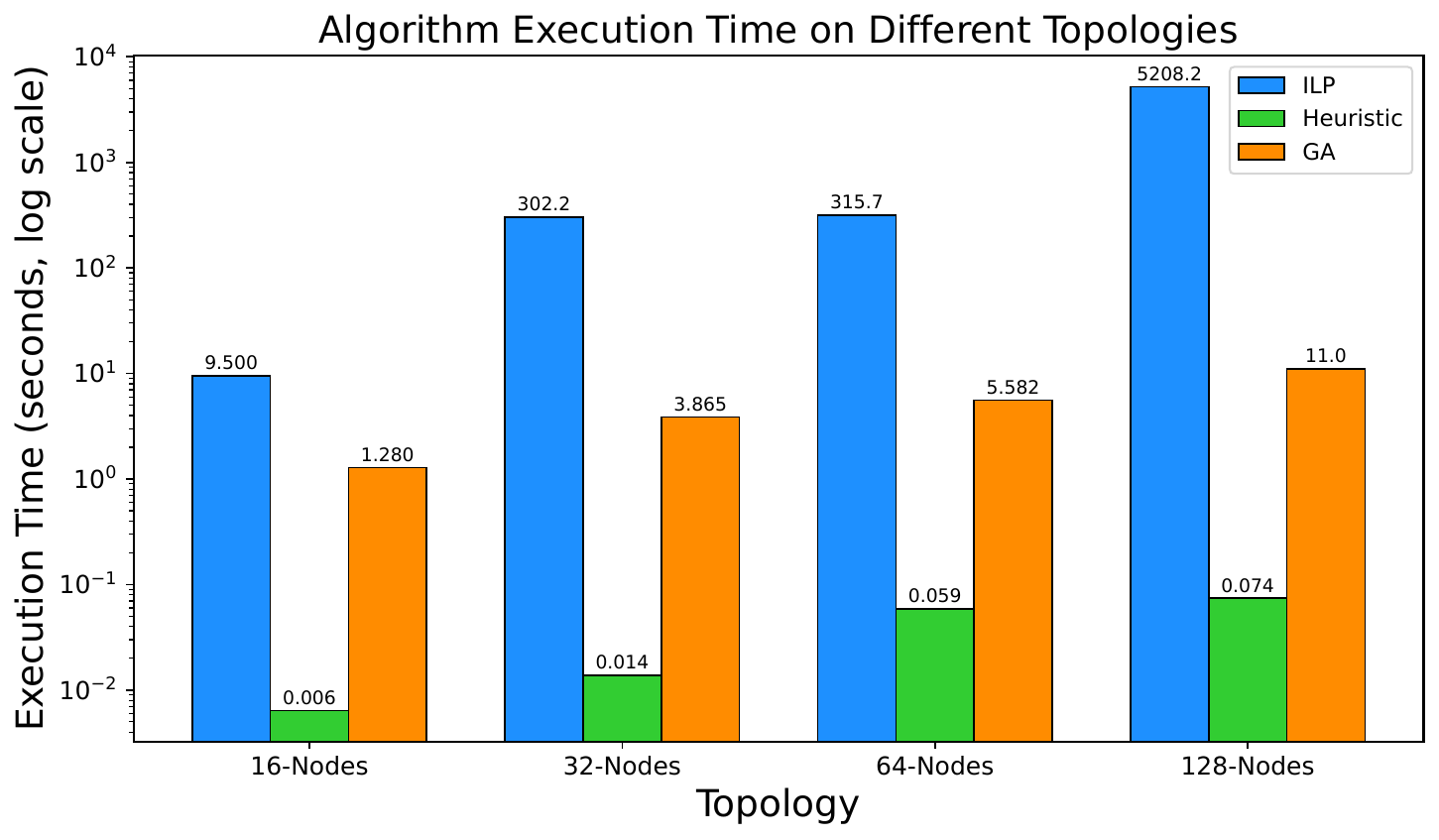}
    \caption{Execution time (in seconds, log scale) of different algorithms across increasing topology sizes. }
    \label{fig:execution-time}
\end{figure}
On the contrary, the heuristic approach show remarkable efficiency in terms of computational time, with execution times in milliseconds even for the larger topologies. For instance, it ran for just 74 milliseconds for even 128-nodes topology, which is still orders of magnitude faster than the ILP approaches. Notably, the GA implementations offer an intermediate solution with execution times ranging from 1.2 seconds for the smaller topologies to around 11 seconds for the 128-node topology. This illustrates a tradeoff between the computational complexity and the solution quality: while the ILP derives an optimal solution, it is at the expense of high computational times and is not scalable with the network size. Thus, their practical applicability is limited to smaller networks or offline planning scenarios. On the other hand, the heuristic and GA approaches, in particular, the GA implementations, emerge as a viable alternative for real-time and near real-time network optimization in larger topologies, where a slight reduction in the solution quality will be more than compensated by their considerable computational efficiency.


\section{Conclusion}
This paper addressed the complex challenge of optimizing virtual radio function placement in disaggregated 5G vRANs while jointly considering optical transport constraints. By integrating slice awareness, reliability, and optical-layer provisioning into a unified framework, our approach bridges the gap between radio and transport resource orchestration. The results demonstrate that coordinated optimization across these domains can significantly enhance profit, resource utilization, and service availability for mobile network operators. Moreover, the proposed shared backup strategy and metaheuristic solutions provide a practical pathway toward scalable and fault-tolerant vRAN deployments. Beyond immediate gains in efficiency and resiliency, this work underscores the importance of cross-layer intelligence for future network automation. Building on these insights, future research will focus on adaptive, AI-driven orchestration and real-time reconfiguration mechanisms tailored to Open RAN and dynamic traffic scenarios, ultimately paving the way for more autonomous and resilient 6G-ready infrastructures.
\ifCLASSOPTIONcaptionsoff
  \newpage
\fi



%
\bibliographystyle{IEEEtran}
\bibliography{references}

%









\end{document}